\documentclass[sigconf]{acmart}

\copyrightyear{2017}
\acmYear{2017}
\setcopyright{acmcopyright}
\acmConference{CIKM'17 }{November 6--10, 2017}{Singapore,
Singapore}
\acmPrice{15.00}
\acmDOI{10.1145/3132847.3132966}
\acmISBN{978-1-4503-4918-5/17/11}

\usepackage{graphicx}
\usepackage{subcaption}
\usepackage{balance}  
\usepackage{amsmath,epsfig}
\usepackage{amssymb}
\usepackage[noend]{algpseudocode}
\usepackage{algorithm}
\usepackage{epstopdf}
\usepackage{multi row}     
\usepackage[font={bf}, tableposition=bottom]{caption}     
\usepackage{microtype}    
\usepackage{units}     
\usepackage{xspace}     
\usepackage{booktabs}     
\usepackage{siunitx}          
\usepackage{xcolor}
\usepackage{tablefootnote}
\usepackage{float}
\usepackage{xfrac}
\usepackage{balance}  

\usepackage{tikz,pgfplots,pgfplotstable}
\usetikzlibrary{positioning,fit,shapes}
\usetikzlibrary{decorations.pathreplacing}
\usetikzlibrary{arrows}

\pgfdeclarelayer{background}
\pgfdeclarelayer{foreground}
\pgfsetlayers{background,main,foreground}

\makeatletter
\tikzset{multicircle/.style  args={#1, #2}{%
 alias=tmp@name, %
  postaction={%
    insert path={
     \pgfextra{%
     \pgfpointdiff{\pgfpointanchor{\pgf@node@name}{center}}%
                  {\pgfpointanchor{\pgf@node@name}{east}}%
     \pgfmathsetmacro\insiderad{\pgf@x}%
        \fill[white] (\pgf@node@name.center)  circle (\insiderad-\pgflinewidth);%
        \draw[#2] (\pgf@node@name.center)  circle (\insiderad-\pgflinewidth);%
        \fill[#2] (\pgf@node@name.center)  -- ++(0:\insiderad-\pgflinewidth) arc (0:#1:\insiderad-\pgflinewidth)--cycle;%
        }}}}}
\makeatother

\definecolor{yafaxiscolor}{rgb}{0.3, 0.3, 0.3}
\definecolor{yafcolor1}{rgb}{0.4, 0.165, 0.553}
\definecolor{yafcolor2}{rgb}{0.949, 0.482, 0.216}
\definecolor{yafcolor3}{rgb}{0.47, 0.549, 0.306}
\definecolor{yafcolor4}{rgb}{0.925, 0.165, 0.224}
\definecolor{yafcolor5}{rgb}{0.141, 0.345, 0.643}
\definecolor{yafcolor6}{rgb}{0.965, 0.933, 0.267}
\definecolor{yafcolor7}{rgb}{0.627, 0.118, 0.165}
\definecolor{yafcolor8}{rgb}{0.878, 0.475, 0.686}
\definecolor{yafcolor9}{rgb}{0.965, 0.733, 0.767}

\newlength{\yafaxispad}
\newlength{\yaftlpad}
\newlength{\yaflabelpad}
\newlength{\yafaxiswidth}
\newlength{\yafticklen}
\makeatletter
\def\pgfplots@drawtickgridlines@INSTALLCLIP@onorientedsurf#1{}
\makeatother

\newcommand{\yafdrawxaxis}[2]{
  \pgfplotstransformcoordinatex{#1}\let\xmincoord=\pgfmathresult 
  \pgfplotstransformcoordinatex{#2}\let\xmaxcoord=\pgfmathresult 
  \pgfsetlinewidth{\yafaxiswidth} 
  \pgfsetcolor{yafaxiscolor}
  \pgfpathmoveto{\pgfpointadd{\pgfpointadd{\pgfplotspointrelaxisxy{0}{0}}{\pgfqpointxy{\xmincoord}{0}}}{\pgfqpoint{-0.5\yafaxiswidth}{\yafaxispad}}}
  \pgfpathlineto{\pgfpointadd{\pgfpointadd{\pgfplotspointrelaxisxy{0}{0}}{\pgfqpointxy{\xmaxcoord}{0}}}{\pgfqpoint{0.5\yafaxiswidth}{\yafaxispad}}}
  \pgfusepath{stroke}

}
\newcommand{\yafdrawyaxis}[2]{
  \pgfplotstransformcoordinatey{#1}\let\ymincoord=\pgfmathresult 
  \pgfplotstransformcoordinatey{#2}\let\ymaxcoord=\pgfmathresult 
  \pgfsetlinewidth{\yafaxiswidth} 
  \pgfsetcolor{yafaxiscolor}
  \pgfpathmoveto{\pgfpointadd{\pgfpointadd{\pgfplotspointrelaxisxy{0}{0}}{\pgfqpointxy{0}{\ymincoord}}}{\pgfqpoint{\yafaxispad}{-0.5\yafaxiswidth}}}
  \pgfpathlineto{\pgfpointadd{\pgfpointadd{\pgfplotspointrelaxisxy{0}{0}}{\pgfqpointxy{0}{\ymaxcoord}}}{\pgfqpoint{\yafaxispad}{0.5\yafaxiswidth}}}
  \pgfusepath{stroke}
}

\pgfplotscreateplotcyclelist{yaf}{%
{yafcolor1,mark options={scale=0.75},mark=o}, 
{yafcolor2,mark options={scale=0.75},mark=square},
{yafcolor3,mark options={scale=0.75},mark=triangle},
{yafcolor4,mark options={scale=0.75},mark=o},
{yafcolor5,mark options={scale=0.75},mark=o},
{yafcolor6,mark options={scale=0.75},mark=o},
{yafcolor7,mark options={scale=0.75},mark=o},
{yafcolor8,mark options={scale=0.75},mark=o}} 

\pgfplotsset{axis y line=left, axis x line=bottom,
  tick align=outside,
  compat = 1.3,
  tickwidth=\yafticklen,
  clip = false,
  every axis title shift = 0pt,
    x axis line style= {-, line width = 0pt, opacity = 0},
    y axis line style= {-, line width = 0pt, opacity = 0},
    x tick style= {line width = \yafaxiswidth, color=yafaxiscolor, yshift = \yafaxispad},
    y tick style= {line width = \yafaxiswidth, color=yafaxiscolor, xshift = \yafaxispad},
    x tick label style = {font=\scriptsize, yshift = \yaftlpad},
    y tick label style = {font=\scriptsize, xshift = \yaftlpad},
    every axis y label/.style = {at = {(ticklabel cs:0.5)}, rotate=90, anchor=center, font=\scriptsize, yshift = -\yaflabelpad},
    every axis x label/.style = {at = {(ticklabel cs:0.5)}, anchor=center, font=\scriptsize, yshift = \yaflabelpad},
    x tick label style = {font=\scriptsize, yshift = 1pt},
    grid = major,
    major grid style  = {dash pattern = on 1pt off 3 pt},
  every axis plot post/.append style= {line width=\yafaxiswidth} ,
  legend cell align = left,
  legend style = {inner sep = 1pt, cells = {font=\scriptsize}},
  legend image code/.code={%
    \draw[mark repeat=2,mark phase=2,#1] 
    plot coordinates { (0cm,0cm) (0.15cm,0cm) (0.3cm,0cm) };%
  } 
}

\newcommand{\spara}[1]{\noindent{\bf #1}}

\newtheorem{problem}[theorem]{Problem}

\newcommand{\graph}{\ensuremath{G}\xspace}
\newcommand{\dgraph}[2]{\ensuremath{G=(#1,#2)}\xspace}
\newcommand{\vertices}{\ensuremath{V}}
\newcommand{\edges}{\ensuremath{E}\xspace}
\newcommand{\subgraph}{\ensuremath{G(S)}\xspace}
\newcommand{\subsetv}{\ensuremath{S}\xspace}
\newcommand{\subsete}[1]{\ensuremath{E(#1)}\xspace}
\newcommand{\numnodes}{\ensuremath{n}\xspace}
\newcommand{\numedges}{\ensuremath{m}\xspace}
\newcommand{\density}[1]{\ensuremath{\rho_{#1}}\xspace}
\newcommand{\estimatedensity}{\ensuremath{\tilde{\rho}_{\densest}}\xspace}

\newcommand{\densest}{\ensuremath{{S}^*}\xspace}
\newcommand{\neighbors}[1]{\ensuremath{N(#1)}\xspace}
\newcommand{\degree}[1]{\ensuremath{d(#1)}\xspace}
\newcommand{\idegree}[2]{\ensuremath{d_{#2}(#1)}\xspace}
\newcommand{\induced}[1]{\ensuremath{\mathcal{I}(#1)}\xspace}
\newcommand{\coreset}{\ensuremath{C}\xspace}

\newcommand{\core}[1]{\ensuremath{\kappa(#1)}\xspace}
\newcommand{\intcore}[2]{\ensuremath{\kappa_{#2}(#1)}\xspace}
\newcommand{\mcore}[1]{\ensuremath{C_{\ell}(#1)}\xspace}
\newcommand{\mcoreg}[1]{\ensuremath{C_{#1}}\xspace}

\newcommand{\streamedge}[1]{\ensuremath{e_{#1}}\xspace}
\newcommand{\streamgraph}[2]{\ensuremath{G_{{#1}}(#2)}\xspace}
\newcommand{\slidingwindow}[2]{\ensuremath{W_{{#1}}(#2)}\xspace}
\newcommand{\slidingvertices}[2]{\ensuremath{V_{{#1}}(#2)}\xspace}
\newcommand{\snowball}{\ensuremath{D}\xspace}
\newcommand{\isnowball}[1]{\ensuremath{D_{#1}}\xspace}
\newcommand{\bag}{\ensuremath{B}\xspace}
\newcommand{\topk}[1]{\ensuremath{\rho_k(#1)}\xspace}
\newcommand{\algo}{\ensuremath{A}\xspace}
\newcommand{\kdensest}{\ensuremath{k}\xspace}
\newcommand{\windowsize}{\ensuremath{x}\xspace}
\newcommand{\collectionsubsetv}{\ensuremath{\mathcal{S}}\xspace} 
\newcommand{\np}{\ensuremath{\mathbf{NP}}\xspace} 
\newcommand{\nphard}{{\np}-hard\xspace} 
\newcommand{\bigO}{\ensuremath{\mathcal{O}}\xspace} 
\newcommand{\polylog}{\ensuremath{\mathrm{polylog}}\xspace} 
\newcommand{\mapreduce}{{Map\-Reduce}\xspace} 

\newcommand{\charikar}{\ensuremath{CH}\xspace}
\newcommand{\batagelj}{\ensuremath{VB}\xspace}
\newcommand{\bahmani}[1]{\ensuremath{BB_{#1}}\xspace}
\newcommand{\li}{\ensuremath{RL}\xspace}
\newcommand{\sariyuce}{\ensuremath{TR}\xspace}
\newcommand{\epasto}[1]{\ensuremath{AE_{#1}}\xspace}
\newcommand{\our}{\ensuremath{GR}\xspace}

\algdef{SE}[DOWHILE]{Do}{doWhile}{\algorithmicdo}[1]{\algorithmicwhile\ #1}%

\usepackage{environ}
\NewEnviron{myequation}{%
    \begin{equation}
    \scalebox{1}{$\BODY$}
    \end{equation}
    }

\newenvironment {squishlist}
{\begin{list}{$\bullet$}
  { \setlength{\itemsep}{1pt}
     \setlength{\parsep}{1pt}
     \setlength{\topsep}{1pt}
     \setlength{\partopsep}{1pt}
     \setlength{\leftmargin}{1.5em}
     \setlength{\labelwidth}{1em}
     \setlength{\labelsep}{0.5em} } }
{\end{list}}

\errorcontextlines\maxdimen

\makeatletter
    \newcommand*{\algrule}[1][\algorithmicindent]{\makebox[#1][l]{\hspace*{.5em}\thealgruleextra\vrule height \thealgruleheight depth \thealgruledepth}}%
\newcommand*{\thealgruleextra}{}
\newcommand*{\thealgruleheight}{.75\baselineskip}
\newcommand*{\thealgruledepth}{.25\baselineskip}

\newcount\ALG@printindent@tempcnta
\def\ALG@printindent{%
    \ifnum \theALG@nested>0
        \ifx\ALG@text\ALG@x@notext
        \else
            \unskip
            \addvspace{-1pt}
            \ALG@printindent@tempcnta=1
            \loop
                \algrule[\csname ALG@ind@\the\ALG@printindent@tempcnta\endcsname]%
                \advance \ALG@printindent@tempcnta 1
            \ifnum \ALG@printindent@tempcnta<\numexpr\theALG@nested+1\relax
            \repeat
        \fi
    \fi
    }%
\usepackage{etoolbox}
\patchcmd{\ALG@doentity}{\noindent\hskip\ALG@tlm}{\ALG@printindent}{}{\errmessage{failed to patch}}
\makeatother

\newbox\statebox
\newcommand{\myState}[1]{%
    \setbox\statebox=\vbox{#1}%
    \edef\thealgruleheight{\dimexpr \the\ht\statebox+1pt\relax}%
    \edef\thealgruledepth{\dimexpr \the\dp\statebox+1pt\relax}%
    \ifdim\thealgruleheight<.75\baselineskip
        \def\thealgruleheight{\dimexpr .75\baselineskip+1pt\relax}%
    \fi
    \ifdim\thealgruledepth<.25\baselineskip
        \def\thealgruledepth{\dimexpr .25\baselineskip+1pt\relax}%
    \fi
    \State #1%
    \def\thealgruleheight{\dimexpr .75\baselineskip+1pt\relax}%
    \def\thealgruledepth{\dimexpr .25\baselineskip+1pt\relax}%
}

\DeclareGraphicsExtensions{.pdf,.png,.jpg,.eps}
\graphicspath{{./img/}}

\begin{document}


\title{Fully Dynamic Algorithm for Top-\kdensest Densest Subgraphs}

%
\author{
Muhammad Anis Uddin Nasir{\small $^{\sharp 1}$}, 
Aristides Gionis{\small $^{\ddagger 2}$}, 
Gianmarco De Francisci Morales{\small $^{\diamond 3}$} \\
Sarunas Girdzijauskas{\small $^{\sharp 4}$}
}
\affiliation{%
  \institution{
$^{\sharp}$Royal Institute of Technology, Sweden\hspace{3mm}
$^{\ddagger}$Aalto University, Finland\hspace{3mm}
$^{\diamond}$Qatar Computing Research Institute, Qatar 
  }
  \vspace{-3mm}
\small{$^{1}$}anisu@kth.se,
\small{$^{2}$}aristides.gionis@aalto.fi,
\small{$^{3}$}gdfm@acm.org,
\small{$^{4}$}sarunasg@kth.se
}


\begin{abstract}
Given a large graph,
the densest-subgraph problem asks to find a subgraph 
with maximum average degree. 
When considering the top-$k$ version of this problem, 
a na\"ive solution is to 
iteratively find the densest subgraph
and remove it in each iteration.  
However, such a solution is impractical due to high processing cost.
The problem is further complicated when dealing with dynamic graphs, 
since adding or removing an edge requires re-running the algorithm.
In this paper, we study the top-$k$ densest-subgraph problem in the sliding-window model and propose an efficient fully-dynamic algorithm.
The input of our algorithm consists of an edge stream, 
and the goal is to find the node-disjoint subgraphs 
that maximize the sum of their densities.
In contrast to existing state-of-the-art solutions 
that require iterating over the entire graph upon any update, 
our algorithm profits from the observation that updates 
only affect a limited region of the graph.
Therefore, the top-$k$ densest subgraphs are maintained by only applying local updates.
We provide a theoretical analysis of the proposed algorithm and show empirically that the algorithm often generates denser subgraphs than state-of-the-art competitors. 
Experiments show an improvement in efficiency of up to five orders of magnitude 
compared to state-of-the-art solutions. 
\end{abstract}

\maketitle
\section{Introduction}

Finding a subgraph with maximal density in a given graph is a fundamental 
graph-mining problem, known as the \emph{densest-subgraph} problem.
Density is commonly defined as the ratio between number of edges and vertices, 
while many other definitions of density have been used in the literature~\cite{batagelj2003,sozio2010community,tatti2015density,tsourakakis2013denser}.
The densest-subgraph problem has many applications, for example, in 
community detetion~\cite{chen2012dense,dourisboure2007extraction}, 
event detection~\cite{angel2012dense}, 
link-spam detection \cite{gibson2005discovering}, and
distance query indexing \cite{akiba2013fast}.

In applications, we are often interested not only in one densest subgraph,
but in the \emph{top-$k$}.
The top-$k$ densest subgraphs 
can be vertex-disjoint, edge-disjoint, or 
overlapping \cite{balalau2015finding, galbrun2016top}.
Different objective functions and constraints give rise to different problem formulations~\cite{galbrun2016top,valari2012discovery, balalau2015finding}.
In this work, we choose to maximize the sum of the densities of the $k$ subgraphs in the solution.
In addition, we seek a solution with disjoint subgraphs.
This version of the problem is known to be NP-hard~\cite{balalau2015finding}.

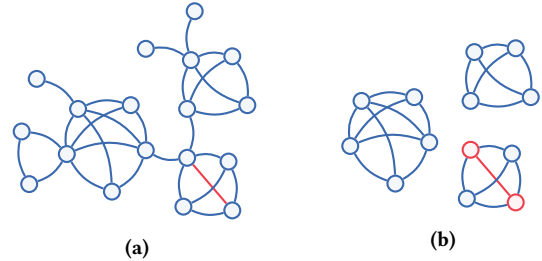
\begin{figure}[t]
\begin{subfigure}{0.48\columnwidth}
\centering 
  \begin{tikzpicture} 

\tikzstyle{vertex} = [font=\small, inner sep = 1pt,  minimum width=6pt, 
	thick, circle, text=black!90, draw=yafcolor5!90, fill=yafcolor5!05]
\tikzstyle{line} = [thick, yafcolor5!90]

\tikzset{
  text style/.style={text=black!70, font=\footnotesize}
}

\node[vertex] (v1) at (1.1,1) {};
\node[vertex] (v2) at (1.55,1.55) {};
\node[vertex] (v3) at (1.35,2.15) {};
\node[vertex] (v4) at (0.65,2.1) {};
\node[vertex] (v5) at (0.5,1.5) {};

\node[vertex] (v6) at (0.0,1.1) {};
\node[vertex] (v7) at (-0.1,1.8) {};
\node[vertex] (v8) at (0.1,2.5) {};

\node[vertex] (v9) at (2.7,0.75) {};
\node[vertex] (v10) at (2.65,1.4) {};
\node[vertex] (v11) at (2.1,1.45) {};
\node[vertex] (v12) at (2.05,0.8) {};

\node[vertex] (v13) at (2.9,2.15) {};
\node[vertex] (v14) at (2.7,2.8) {};
\node[vertex] (v15) at (2.15,2.75) {};
\node[vertex] (v16) at (2.1,2.1) {};

\node[vertex] (v17) at (1.55,2.9) {};
\node[vertex] (v18) at (2.2,3.4) {};

\draw (v1) edge [line, -, bend right = 20] (v2);
\draw (v1) edge [line, -, bend right = 20] (v4);
\draw (v1) edge [line, -, bend left = 20] (v5);
\draw (v2) edge [line, -, bend right = 20] (v3);
\draw (v2) edge [line, -, bend right = 20] (v4);
\draw (v2) edge [line, -, bend right = 20] (v5);
\draw (v3) edge [line, -, bend right = 20] (v4);
\draw (v3) edge [line, -, bend right = 20] (v5);
\draw (v4) edge [line, -, bend right = 20] (v5);

\draw (v5) edge [line, -, bend left = 20] (v6);
\draw (v5) edge [line, -, bend right = 20] (v7);
\draw (v6) edge [line, -, bend left = 20] (v7);

\draw (v4) edge [line, -, bend right = 20] (v8);

\draw (v9) edge [line, -, bend right = 20] (v10);
\draw (v9) edge [line, yafcolor4!90, -, bend left = 0] (v11);
\draw (v9) edge [line, -, bend left = 20] (v12);
\draw (v10) edge [line, -, bend right = 20] (v11);
\draw (v10) edge [line, -, bend left = 20] (v12);
\draw (v11) edge [line, -, bend right = 20] (v12);

\draw (v2) edge [line, -, bend right = 20] (v11);
\draw (v11) edge [line, -, bend right = 20] (v16);

\draw (v13) edge [line, -, bend right = 20] (v14);
\draw (v13) edge [line, -, bend left = 20] (v15);
\draw (v13) edge [line, -, bend left = 20] (v16);
\draw (v14) edge [line, -, bend right = 20] (v15);
\draw (v14) edge [line, -, bend right = 20] (v16);
\draw (v15) edge [line, -, bend right = 20] (v16);

\draw (v15) edge [line, -, bend right = 20] (v17);
\draw (v15) edge [line, -, bend left = 20] (v18);

\end{tikzpicture}
  \caption{}
  \label{fig:sample-graph}
\end{subfigure}%
\begin{subfigure}{0.48\columnwidth}
  \centering
  \begin{tikzpicture} 

\tikzstyle{vertex} = [font=\small, inner sep = 1pt,  minimum width=6pt, 
	thick, circle, text=black!90, draw=yafcolor5!90, fill=yafcolor5!05]
\tikzstyle{line} = [thick, yafcolor5!90]

\tikzset{
  text style/.style={text=black!70, font=\footnotesize}
}

\node[vertex] (v1) at (1.1,1) {};
\node[vertex] (v2) at (1.55,1.55) {};
\node[vertex] (v3) at (1.35,2.15) {};
\node[vertex] (v4) at (0.65,2.1) {};
\node[vertex] (v5) at (0.5,1.5) {};


\node[vertex, draw=yafcolor4!90, fill=yafcolor4!05] (v9) at (2.7,0.75) {};
\node[vertex] (v10) at (2.65,1.4) {};
\node[vertex, draw=yafcolor4!90, fill=yafcolor4!05] (v11) at (2.1,1.45) {};
\node[vertex] (v12) at (2.05,0.8) {};

\node[vertex] (v13) at (2.9,2.15) {};
\node[vertex] (v14) at (2.7,2.8) {};
\node[vertex] (v15) at (2.15,2.75) {};
\node[vertex] (v16) at (2.1,2.1) {};


\draw (v1) edge [line, -, bend right = 20] (v2);
\draw (v1) edge [line, -, bend right = 20] (v4);
\draw (v1) edge [line, -, bend left = 20] (v5);
\draw (v2) edge [line, -, bend right = 20] (v3);
\draw (v2) edge [line, -, bend right = 20] (v4);
\draw (v2) edge [line, -, bend right = 20] (v5);
\draw (v3) edge [line, -, bend right = 20] (v4);
\draw (v3) edge [line, -, bend right = 20] (v5);
\draw (v4) edge [line, -, bend right = 20] (v5);



\draw (v9) edge [line, -, bend right = 20] (v10);
\draw (v9) edge [line, yafcolor4!90, -, bend left = 0] (v11);
\draw (v9) edge [line, -, bend left = 20] (v12);
\draw (v10) edge [line, -, bend right = 20] (v11);
\draw (v10) edge [line, -, bend left = 20] (v12);
\draw (v11) edge [line, -, bend right = 20] (v12);


\draw (v13) edge [line, -, bend right = 20] (v14);
\draw (v13) edge [line, -, bend left = 20] (v15);
\draw (v13) edge [line, -, bend left = 20] (v16);
\draw (v14) edge [line, -, bend right = 20] (v15);
\draw (v14) edge [line, -, bend right = 20] (v16);
\draw (v15) edge [line, -, bend right = 20] (v16);


\end{tikzpicture}
  \caption{}
  \label{fig:sample-bag}
  	\vspace{-\baselineskip}
\end{subfigure}


\caption{
For the graph in Figure \ref{fig:sample-graph}, 
we are interested in extracting the top-3 densest subgraphs.
Consider the arrival of an edge shown in red. 
Figure \ref{fig:sample-bag} shows the top-3 densest subgraphs 
after the arrival.
The objective is to design an algorithm that can efficiently
maintain the densest subgraphs while keeping the number of updates very low, 
in this case updating only the vertices in red.}
\label{fig:sample-figure}
\vspace{-\baselineskip}
\end{figure}

To complicate the matter, most real-world graphs are dynamic and rapidly changing.
For instance, Facebook users are continuously creating new connections and removing old ones, thus changing the network structure.
Twitter users produce posts at a high rate, which makes old posts less relevant.
Given the dynamic nature of many graphs, here we focus on a sliding-window model which gives more importance to recent events~\cite{babcock2002models,crouch2013dynamic,datar2002maintaining}.
Finding the top-$k$ densest subgraphs in a sliding window is of interest to several real-time applications, e.g.,  
community tracking~\cite{yang2015defining}, 
event detection~\cite{rozenshtein2014event}, 
story identification~\cite{angel2012dense},
fraud detection~\cite{beutel2013copycatch}, and more.
%
%
%
We assume the input to the system arrives as an \emph{edge stream}, 
and seek to extract the $k$ vertex-disjoint subgraphs that maximize the sum of densities~\cite{balalau2015finding}.

A na\"{i}ve solution involves executing a static algorithm for the densest-subgraph problem $k$ times, while removing the densest subgraph in each iteration.  
However, such a solution is impractical as it requires to execute the algorithm $k$ times for each update.
An alternative solution to our problem is to use a dynamic densest subgraph algorithm in a pipeline manner, 
where the output of an algorithm instance serves as input to the following one.
In this case, the graph and the instances of the algorithm are replicated independently across 
$k$ instances of the algorithm, resulting in a high memory and processing cost.

In this paper, we propose a fully-dynamic algorithm that finds an approximate solution.
The proposed algorithm follows a greedy approach and updates the densities of the subgraphs connected to vertices affected by edge operations (addition and removal).
The algorithm is efficiently designed based on key properties of densest subgraphs, 
and it is competitive against other recent dynamic algorithms~\cite{bhattacharya2015space, epasto2015efficient, mcgregor2015densest}.

First, our algorithm relies on the observation 
that only high-degree vertices are relevant for the solution.
As many natural graphs have a heavy-tailed degree distribution, 
the number of high-degree vertices in a graph is relatively smaller 
than the number of low-degree ones.
This simple observation enables pruning a major portion of the input 
stream on-the-fly.
Second, the vertices that are part of a densest subgraph are connected strongly to each other and weakly to other parts of the graph.
This enables independently maintaining and locally updating multiple subgraphs. 
Figure~\ref{fig:sample-figure} provides an example which demonstrates this intuition.

The algorithm tracks multiple subgraphs on-the-fly with the help of a newly defined data structure called \emph{snowball}.
These subgraphs are stored in a bag, from which the $k$ subgraphs with maximum densities are extracted.
The algorithm runs in-place, and does not require multiple copies of the graph, thus making it memory-efficient.
The one-pass nature of the algorithm allows extracting top-$k$ densest subgraphs for larger values of $k$.

We provide a theoretical analysis of the proposed algorithm, and show that the algorithm guarantees $2$-approximation for the first densest subgraph ($k=1$) while providing a high-quality heuristic for $k > 1$ compared to other solutions.
Experimental evaluation shows that our algorithm often generates denser subgraphs compared to the state-of-the-art algorithms, due to the fact that it maintains disconnected subgraphs separately.
In addition, the algorithm provides improvement in runtime up to three to five orders of magnitude compared to the state-of-the-art.
In summary, we make the following contributions:

\begin{squishlist}
\item We study the top-$k$ densest vertex-disjoint subgraphs problem in the sliding-window model.
\item We provide a brief survey on adapting several algorithms for densest subgraph problem for the top-$k$ case.
\item We propose a scalable fully-dynamic algorithm for the problem, and provide a detailed analysis of it.
\item The algorithm is open source and available online, together with the implementations of all the baselines.\footnote{\url{https://github.com/anisnasir/TopKDensestSubgraph}}
\item We report a comprehensive empirical evaluation of the algorithm in which it significantly outperforms previous state-of-the-art solutions by several orders of magnitude, while producing comparable or better quality solutions.
\end{squishlist}

\section{Preliminaries}
\label{sec:prel}

In this section, we present our notation, 
revisit basic definitions,
and formulate the top-$k$ densest subgraphs problem.

Consider an undirected graph $\dgraph{\vertices}{\edges}$
with \numnodes $=|\vertices|$ vertices and \numedges$=  |  \edges | $ edges.
The {\em neighborhood} of $v\in \vertices$  is defined as 
$\neighbors{v} =\{u \mid (v,u)\in \edges\}$, 
and its {\em degree} as $ \degree{v} =| \neighbors{v} |$.
For a subset $\subsetv\subseteq \vertices$
we define $\subsete{\subsetv}$ to be the set of edges
whose both endpoints are in~\subsetv, 
and $\subgraph=(\subsetv,\subsete{\subsetv})$ the subgraph {\em induced} by $\subsetv$.
The {\em internal degree} of a vertex $v$ 
with respect to $\subsetv\subseteq\vertices$ 
is defined by $\idegree{v}{\subsetv} = |\neighbors{v} \cap \subsetv|$.
Finally, for a subset of vertices $\subsetv \subseteq \vertices$
we define its {\em density} $\density{\subsetv}$  by
\begin{equation}
\label{equation:density}
\density{\subsetv}=\dfrac{|\subsete{\subsetv}|}{ | \subsetv |}.
\end{equation}
Note that the density of any subgraph is equal to half of its \emph{average internal degree}. 


\begin{definition} [Densest subgraph]
\label{definition:densest}
Given an undirect\-ed graph $\dgraph{\vertices}{\edges}$,
the densest subgraph \densest is a set of vertices that
maximizes the density function, i.e., 
\begin{equation}
\label{equation:densest}
\densest = \arg\max_{\subsetv\subseteq\vertices} \density{\subsetv}.
\end{equation}
\end{definition}

We say that an algorithm {\algo} computes an $\alpha$-{\em approximation} of the densest subgraph
if {\algo} computes a subset $\subsetv \subseteq \vertices$ 
such that $\density{\subsetv} \geq  \frac{1}{\alpha} \density{\densest}$, 
where $\densest \subseteq \vertices$ is the densest subgraph of \graph.


Next we introduce other concepts related to densest subgraph: 
{\em graph core}, 
{\em core decomposition}, and 
{\em induced core subgraph} of a vertex.


\begin{definition} [$j$-core]
\label{definition:k-core}
Given an undirected graph $\dgraph{\vertices}{\edges}$
and an integer $j$, 
a $j$-core of {\graph} is 
a subset of vertices $\coreset\subseteq\vertices$
so that each vertex $v\in\coreset$ has internal degree $\idegree{v}{\coreset}\ge j$, 
and \coreset is maximal with respect to this property.
\end{definition}

\begin{definition} [Core decomposition]
\label{definition:k-core-decomposition}
A core decomposition of a graph $\dgraph{\vertices}{\edges}$
is a nested sequence $\big\{\coreset_i\big\}$ of cores 
\begin{equation}
\label{equation:k-core-decomposition}
\vertices = \coreset_{0} \supseteq \coreset_{1} \supseteq \ldots 
\supseteq \coreset_{\ell} \supseteq \varnothing,  
\end{equation}
where each $\coreset_{i}$ is a $j$-core for some $j$.
\end{definition}

\begin{definition} [Core number]
\label{definition:core-number}
Given a core decomposition 
$\vertices = \coreset_{0} \supseteq \coreset_{1} \supseteq \ldots \supseteq \coreset_{\ell} \supseteq \varnothing$
of a graph \dgraph{\vertices}{\edges},
the core number $\core{v}$ of a vertex $v$
is the largest $j$ such that $v\in\coreset$ and {\coreset} is a $j$-core.
By overwriting notation, 
the core number $\core{\coreset}$ of a core $\coreset$ 
is the largest $j$ for which $\coreset$ is a $j$-core.
\end{definition}

Additionally, we use $\intcore{v}{\subsetv}$
to denote the core number of a vertex~$v$ in the subgraph induced by {\subsetv}.
The {largest core} (or {\em main core}) 
of a subgraph of $\graph(\subsetv)=(\subsetv,\edges(\subsetv))$ 
is denoted by $\mcore{\subsetv}$, 
while the main core of \graph
is simply denoted by $\mcoreg{\ell}$.

Note that the density of a $j$-core is at least~$\sfrac{j}{2}$,
as each vertex in the core has degree at least $j$ and each edge is counted twice.
This observation implies that the main core of a graph 
is a 2-approximation of its densest subgraph. 

\begin{lemma}
\label{lema:2-approximation}
Consider the core decomposition of a graph \graph, i.e., 
$\mcoreg{1} \subseteq \mcoreg{2} \subseteq\ldots \subseteq \mcoreg{\ell}$. 
The maximum core $\mcoreg{\ell}$
is 2-approximation of the densest subgraph of \graph~\cite{kortsarz1994generating}.
\end{lemma}

\begin{proof}
Let \densest be the densest subgraph of {\graph} having density~$\density{\densest}$.
Every vertex in $\graph(\densest)$ has degree at least $\density{\densest}$;
otherwise a vertex with degree smaller than $\density{\densest}$ can be removed to obtain 
an even denser subgraph. 
Thus, \densest is a $\density{\densest}$-core.
Given the core decomposition of \graph, 
we know that $\density{\mcoreg{\ell}} \ge \frac{1}{2}\core{\mcoreg{\ell}}$.
We want to show that $\density{\mcoreg{\ell}}\ge \frac{1}{2}\density{\densest}$.
Assume otherwise, i.e., $\density{\mcoreg{\ell}}<\frac{1}{2}\density{\densest}$.
Then $\core{\mcoreg{\ell}}<\density{\densest}$.
It follows that \densest is a higher-order core than $\mcoreg{\ell}$, 
a contradiction.
\end{proof}


The concept of a core subgraph $\induced{v}$ 
induced by a vertex~$v$ \cite{li2014efficient,sariyuce2013streaming}
is also pertinent to our analysis. 
\begin{definition} [Induced core subgraph]
\label{definition:induced-core}
Given a graph $G=(V,E)$ and a vertex 
$v\in \vertices$, the induced core subgraph of vertex $v$, 
denoted by $\induced{v}$, 
is a maximal connected subgraph containing the vertex $v$ 
s.t. the core number of all the vertices in $\induced{v}$ 
is equal to $\core{v}$.
\end{definition}
In other words, the induced subgraph contains all vertices that are reachable from $v$ 
and have the same core number $\core{v}$.

All previous definitions apply to static graphs. 
Let us now focus on dynamic graphs.
In particular, we consider processing a graph in the 
\emph{sliding window edge-stream model}~\cite{datar2002maintaining}.
According to this model, 
the input to our problem is a stream of edges. 
The edge $\streamedge{i}$ is the $i$-th element of the stream. 
Equivalently, we say that edge $\streamedge{i}$ has timestamp~$i$.
A sliding window $\slidingwindow{t}{\windowsize}$, defined at time $t$ and of size {\windowsize}, 
is the set of all edges that arrive
between $\streamedge{t-\windowsize+1}$ and $\streamedge{t}$,
\begin{myequation}
\slidingwindow{t}{\windowsize} = \{ \streamedge{i}, i \in [t-\windowsize+1,t]\}.
\end{myequation}
For each edge $\streamedge{i} = (u, v)$,
we consider that $u$ and $v$ appear at time $i$, 
and we use $\slidingvertices{t}{\windowsize}$ to denote the set of vertices that appear 
in a length-{\windowsize} sliding window at time $t$.
The graph in a length-{\windowsize} sliding window at time $t$ is then defined to be 
$\streamgraph{t}{\windowsize}=(\slidingvertices{t}{\windowsize}, \slidingwindow{t}{\windowsize})$.

We are now ready to formally define the problem that we consider in this paper, i.e., 
finding the top-$k$ densest subgraphs in sliding window.
We first define the problem in a static setting. 

\begin{definition} 
\label{definition:top-k-densest-static} 
Given an undirected graph \dgraph{\vertices}{\edges} and an integer $k>0$, 
the top-$k$ densest subgraphs of {\graph} is a set of $k$ disjoint maximal set of vertices 
$\collectionsubsetv = \{ \subsetv_1, \ldots, \subsetv_k\}$ that maximize the 
sum of its densities:
\begin{eqnarray}
&& \topk{\collectionsubsetv}  = \smash{\displaystyle\max} \sum \limits_{i=1}^{k}  \density{\subsetv_i}, 
\text{ for all }\subsetv_i \in \collectionsubsetv \quad \text{ subject to }  \nonumber \\
&& \text{there is no } \subsetv_j \supset \subsetv_i \mid \density{\subsetv_j} \geq \density{\subsetv_i}, \quad \text{ for all } \subsetv_i , \subsetv_j \subseteq \collectionsubsetv \\ 
&& \subsetv_i \cap \subsetv_j = \varnothing,  \quad \text{ for all } i, j \in \{1\ldots k\}, i \neq j.
\end{eqnarray}
\end{definition}

As already shown by~\citet{balalau2015finding},
the problem defined aboce is {\nphard}, for any $k>1$.
The problem we consider in this paper is the following. 

\begin{problem} 
\label{problem:top-k-densest-streaming} 
Given a graph stream $\{\streamedge{i}\}$
and a sliding window length {\windowsize}, 
maintain the top-$k$ densest subgraphs 
$\topk{\collectionsubsetv}$ 
of the graph~$\streamgraph{t}{\windowsize}$, at any given time $t$.
\end{problem}

\section{Background}
\label{sec:background}

In this section we present 
a brief review over several algorithms
for finding dense subgraphs. 
Additionally, we discuss how these methods can be used for solving 
Problem~\ref{problem:top-k-densest-streaming}.

\spara{Densest subgraph in static graphs.}
Finding
the densest subgraph according to the density definition~(\ref{equation:density})
can be solved in polynomial time. 
An elegant solution involving reduction to the minimum-cut problem 
was given by \citet{goldberg1984finding}.
As the fastest algorithm to solve the minimum-cut problem runs in 
$\bigO(nm)$ time~\cite{orlin2013maxflows}, 
Goldberg's algorithm is not scalable to large graphs.
 
\citet{asahiro1996greedily} and \citet{charikar2000greedy} propose  
a linear-time algorithm that provides a factor-2 approximation.
This algorithm iteratively removes the vertex 
with the lowest degree in each iteration, 
until left with an empty graph. 
Among all subgraphs considered during this vertex-removal process, 
the algorithm returns the densest.
The time complexity of this greedy algorithm is $\bigO(m+n)$.
\citet{bahmani2012densest} 
propose a \mapreduce version of the greedy algorithm, 
with approximation ratio $2(1+\epsilon)$, 
while making $\bigO(\log_{1+\epsilon}n)$ passes over the input graph.

\spara{Core decomposition in static graphs.}
\label{section:core-decomposition}
The core decomposition of a graph {\graph} 
is the process of identifying all cores of \graph, 
as defined in~\ref{definition:k-core-decomposition}. 
\citet{batagelj2003} propose a linear-time algorithm to obtain the core decomposition.
The algorithm first considers the whole graph 
and then repeatedly removes the vertex with the smallest degree.
The core number $\core{v}$ of a vertex $v$ 
is set equal to the degree of $v$ 
at the moment that $v$ is removed from the graph.


\spara{Densest subgraph in evolving graphs.}
There is a growing body of literature
on finding dense subgraphs in evolving graphs \cite{bhattacharya2015space, mcgregor2015densest, epasto2015efficient, gionis2015dense}. 
We focus mainly on the deterministic algorithm for densest subgraph in evolving graphs. 
For instance, \citet{epasto2015efficient} propose 
an efficient algorithm for computing the densest subgraph in the dynamic graph model~\cite{eppstein1998dynamic}.
Their work assumes that edges are inserted into the graph adversarially but deleted randomly. 
Even though the algorithm can, in practice, handle arbitrary edge deletions, its approximation guarantees hold only under the random edge-deletion assumption.
The algorithm is similar to the one by \citet{bahmani2012densest}, and 
it provides a $2(1+\epsilon)^6$-approximation of the densest subgraph, while requiring 
polylogarithmic
amortized cost per update with high probability.

\spara{Core maintenance in evolving graphs.}
\citet{sariyuce2013streaming} propose the \emph{traversal algorithm}, 
for efficient core maintenance.
This algorithm identifies a small set of vertices 
that are affected by edge updates and processes these vertices in linear time
in order to maintain a valid core decomposition. 
\citet{li2014efficient} propose an efficient three-stage algorithm 
for core maintenance in large dynamic graphs. 
The algorithm maintains a core decomposition of an evolving graph 
by applying updates to very few vertices in the graph.
Once these few vertices have been identified, 
the algorithm computes the correct core numbers via a quadratic operation.


\spara{Finding top-$k$ densest subgraphs.}
 \label{subsection: topk}
The problem of finding top-$k$ densest subgraphs
has been mainly studied for finding 
\emph{overlapping} subgraphs in static graphs~\cite{balalau2015finding,galbrun2016top}.

Next, we discuss how the algorithms
presented above 
(\citet{charikar2000greedy}, \citet{batagelj2003}, \citet{bahmani2012densest}, \citet{sariyuce2013streaming}, 
\citet{li2014efficient}, and \citet{epasto2015efficient})
can be used to produce top-$k$ densest subgraphs.

Our first observation 
is that a set of $k$ dense subgraphs
can be obtained from any algorithm that finds the densest subgraph
by $k$ repeated invocations. 
The time complexity of computing a set of $k$ dense subgraphs in this manner
is simply the running-time complexity of the densest-subgraph algorithm multiplied by $k$.
From a practical point of view, 
all the static algorithms mentioned
are not in-place algorithms, and thus require copying the whole graph for processing.
Furthermore, when a vertex or edge is added or deleted from the graph, 
the whole $k$ dense subgraph computation has to be repeated. 

The second observation is that, 
by using the algorithm of \citet{epasto2015efficient}, 
we can obtain a set of $k$ dense subgraphs by running $k$ instances of the fully-dynamic algorithm 
in a pipeline manner. 
The idea is to run $k$ instances of the algorithm
in which the output of  each instance $i \in \{0\ldots k-1\}$ 
is fed into the next $(i+1)$ instance as a removal operation.
The pipeline version of the algorithm requires keeping $k$ copies of the input graph 
and an additional $\bigO(kn)$ size space for bookkeeping.
Note that the output of each instance of the pipeline might cascade, 
which requires updating the vertices in all the instances. 
In particular, vertices that cease to be part of solution in upstream instances need to be added in downstream instances. 
Likewise the vertices that become part of densest subgraphs in upstream instances need to be removed from a downstream instances.
The modification of the algorithm,
as discussed above, is expensive in terms of memory, 
as it requires replicating the graph and the algorithm's structures $k$ times.
In addition, running and maintaining $k$ parallel instances makes the algorithm compute-intensive.

Finally, 
to maintain top-$k$ densest subgraphs in evolving graphs, 
we can leverage algorithms for core decomposition 
maintenance~\cite{sariyuce2013streaming,li2014efficient}. by leveraging Lemma~\ref{lema:2-approximation},
Thus, the idea is to find and maintain top-$k$ disjoint maximum cores.
In order to maintain such cores
we run a single instance of the algorithm by \citet{sariyuce2013streaming} 
or \citet{li2014efficient} that maintains the core number of all the vertices in the graph.
We then extract the top-$k$ densest subgraphs by:
($i$) extracting the main core, 
($ii$) removing the vertices in the main core and updating the core number for rest of the vertices, and 
($iii$) repeating the steps until $k$ subgraphs are extracted.
\vspace{-4mm}
\section{Algorithm}
The main idea of our algorithm is to maintain and update multiple dense subgraphs online.
These subgraphs are candidates for the top-$k$ densest subgraphs.
However, maintaining multiple subgraphs for fully-dynamic streams requires answering two interesting questions: ($i$) how to reduce the search space of the solution, and ($ii$) how to split the whole graph into subgraphs.

To answer the aforementioned questions, we make two observations. 
First, since dense subgraphs are formed by relatively high-degree vertices, one can find dense subgraphs by keeping track of these high-degree vertices only. 
Second, these subgraphs can be updated locally upon edge updates, without affecting the other parts of the graph.

Based on these observations, we develop an algorithm that reduces the solution space by considering only high-degree vertices, and divides the whole graph into smaller subgraphs, each representing a dense subgraph.
The top-$k$ densest subgraphs among the candidate subgraphs provide a solution for Problem~\ref{problem:top-k-densest-streaming}. 



We begin by designing an algorithm to find the densest subgraph (top-1) and then we extend it to find the top-$k$ densest subgraphs.
Our algorithm might not be the most efficient solution for the (top-1) densest-subgraph problem per se, but it provides efficient outcomes when extended to solve the top-$k$ densest-subgraph problem.

%
We start by defining some properties of the densest subgraph that we leverage in our algorithm.

\begin{lemma} 
\label{lemma-degree}
Given an undirected graph $\dgraph{\vertices}{\edges}$, the densest subgraph $\densest \subseteq \vertices$ with density $\density{\densest}$, all the vertices $v \in \densest$ have degree $\idegree{v}{\densest} \geq \density{\densest}$.
\end{lemma}

\begin{proof}
This lemma holds according to the definition of optimal density. 
In an optimal solution, each vertex has degree larger than or equal to $\density{\densest}$. 
Otherwise, removing the vertex from the subgraph will increase the average degree, and thus the density, of the subgraph.
\end{proof}

Given Lemma \ref{lemma-degree}, 
at any time $t$, 
the densest subgraph $\densest_t$ of graph $\graph_t$ 
contains only vertices $v$ that have degree 
$\degree{v} \geq \idegree{v}{\densest_t} \geq \density{\densest_t}$.
Then, given $\graph_t$ and $\density{\densest_t}$, we want to compute the densest subgraph after the addition of a new vertex $u \notin V_t$ at time $t+1$. 

Let $\degree{u}$ be the degree of vertex $u \in V_{t+1}$ and $\densest_{t+1}$ be the densest subgraph at time $t+1$.
For simplicity, assume that the graph $\graph_{t+1}$ is connected.
According to Lemma \ref{lemma-degree}, for any vertex $u$ to be the part of the densest subgraph, its internal degree satisfies 
$\idegree{u}{\densest_{t+1}}\geq \density{\densest_{t+1}}$.
As vertex $u$ is added to the graph the new density is always greater, i.e., $\density{\densest_{t+1}} \geq \density{\densest_t}$.
Therefore, for vertex $u$ to be the part of densest subgraph, the degree of vertex u should satisfy $\degree{u} \geq \density{\densest_t}$.
Therefore, if the degree of vertex $u$ is lower than the $\density{\densest_t}$, it cannot be part of the densest subgraph $\density{\densest_{t+1}}$ and can be ignored.

Now, considering the case when $\degree{u} \geq \density{\densest_t}$.
Adding the vertex $u$ to densest subgraph will update the density:
\begin{myequation}
\density{\densest_{t+1}} \leftarrow \frac{\mid \subsete{\densest_t} \mid + \ \idegree{u}{\densest_t} }{ \mid \densest_t \mid +1 }.
\end{myequation} 
We also know that $\density{\densest_{t+1}} \geq \density{\densest_t}$, which means that 
\begin{myequation}
 \frac{\mid \subsete{\densest_t} \mid + \ \idegree{u}{\densest_t} }{ \mid \densest_t \mid +1 } \geq \frac{\mid \subsete{\densest_t} \mid}{ \mid \densest_t \mid}.
\end{myequation}

From this inequality it follows $\idegree{u}{\densest_t} \geq \density{\densest_t}$.
Using these properties, we ignore the vertices of the new edge that have degree lower than the current estimate of the density. 
Further, we are interested in finding the main core in the remaining subgraph of high-degree vertices, as it represents a 2-approximation of the densest subgraph according to Lemma~\ref{lema:2-approximation}.
To this end, we propose a new data structure that relies on Lemma~\ref{lemma-degree}, the \emph{snowball}.

\subsection{Snowball}
\label{sec: snowball}
A snowball \snowball is an incremental data structure that stores a strongly connected subgraph,
which maintains the following invariants:
\begin{squishlist}
\item The core number $\intcore{v}{\snowball}$ of each vertex $v \in \snowball$ inside a snowball is equal to the main core ($\mcore{\snowball}$) of the snowball.
\item All the vertices in the snowball are connected.
\end{squishlist}

These invariants ensure that all the vertices in the snowball have the same core number, which is the main core of the snowball by definition.
A snowball maintains these invariants while handling the following graph update operations: a) adding/removing a vertex, and b) adding/removing an edge.


\subsection{Bag of snowballs}
The high-degree vertices in the graph are assigned to a snowball.
As these vertices are not strongly connected, they might end up in different snowballs.
We store each of these disconnected snowballs in a data structure called the bag, denoted by \bag.

The bag ensures that each snowball is vertex disjoint. 
Further, the bag provides an additional operation: extracting the densest snowball among the set of snowballs.
The density of the extracted snowball is the maximal density, which is the threshold separating the high-degree vertices from the low-degree ones.
We denote this estimate of the maximal density \estimatedensity.

The bag is a supergraph which contains a set of snowballs and all the edges between the snowballs.
We maintain all the core numbers of the nodes in the bag by leveraging a core decomposition algorithm (see Section~\ref{section:core-decomposition}). 
The core number of each node in the bag is used to ensure that each node has the maximum possible core number.
Figure \ref{bag-example-1} provides an example explaining one of the issues that may arise.
In the example, the bag contains two snowballs, however, it is possible to produce a new snowball with a larger core number.
\begin{figure}[t]
	\centering
	\begin{tikzpicture}  [scale=1.1,every node/.style={scale=1.2}, shift=(current page.center)]]

\tikzstyle{vertex} = [font=\large, inner sep = 1pt,  minimum width=6pt, 
	thick, circle, text=black!90, draw=yafcolor5!90, fill=yafcolor5!05]
\tikzstyle{line} = [thick, yafcolor5!90]

\tikzset{
  text style/.style={text=black!70, font=\footnotesize}
}

\node[text width=4cm, font=\scriptsize, text centered] at (1.5,1.5) 
    {$D_1$};
\node[text width=4cm, font=\scriptsize, text centered] at (3.5,1.5) 
    {$D_2$};

\node[text width=4cm, font=\scriptsize, text centered] at (2.5,2.3) 
    {Bag};

\draw[thick,rounded corners=8pt] (0.8,0.9) -- (0.8,2.15) -- (2.2,2.15) -- (2.2,0.9) -- (0.8,0.9);
\draw[thick,rounded corners=8pt] (2.8,0.9) -- (2.8,2.15) -- (4.2,2.15) -- (4.2,0.9) -- (2.8,0.9);
\draw[thick,rounded corners=8pt] (0.7,0.8) -- (4.3,0.8) -- (4.3,2.5) -- (0.7,2.5) -- (0.7,0.8);

\node[vertex] (v1) at (1.1,1.1) {};
\node[vertex] (v2) at (1.9,1.1) {};
\node[vertex] (v3) at (1.1, 1.9) {};
\node[vertex] (v4) at (1.9,1.9) {};

\node[vertex] (v5) at (3.1,1.1) {};
\node[vertex] (v6) at (3.9,1.1) {};
\node[vertex] (v7) at (3.1,1.9) {};
\node[vertex] (v8) at (3.9,1.9) {};

\draw (v1) edge [line, -, bend right = 20] (v2);
\draw (v1) edge [line, -, bend left = 20] (v3);
\draw (v2) edge [line, -, bend right = 20] (v4);
\draw (v3) edge [line, -, bend left = 20] (v4);
\draw (v5) edge [line, -, bend right = 20] (v6);
\draw (v5) edge [line, -, bend left = 20] (v7);
\draw (v6) edge [line, -, bend right = 20] (v8);
\draw (v7) edge [line, -, bend left = 20] (v8);

\draw (v4) edge [line, dashed ,yafcolor4!90, -, bend left = 0] (v5);
\draw (v4) edge [line, dashed, yafcolor4!90, -, bend left = 0] (v7);
\draw (v2) edge [line, dashed, yafcolor4!90, -, bend left = 0] (v5);
\draw (v2) edge [line, dashed, yafcolor4!90, -, bend left = 0] (v7);
\end{tikzpicture} 
\caption{Example showing that the bag requires maintaining the core number of the vertices. 
Initially, the bag contains two snowballs with core number 2, i.e., $D_1$ and $D_2$.
Consider the arrival of the edges shown in red. 
The greedy assignment of the edges might skip creating a new snowball with core number 3, using the four nodes in the middle.
}
\label{bag-example-1}
\vspace{-\baselineskip}
\end{figure}
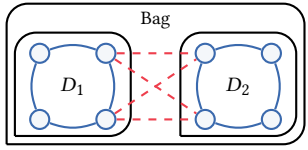
Next, we define the algorithms to update this data structure upon graph updates.

\subsection{Addition operations}

\spara{Vertex addition:}
As discussed in Section \ref{sec:prel}, the updates appear in the form of an edge stream.
Here, we define the vertex addition algorithm that acts as a helper for edge addition.
The algorithm is triggered when at least one of the endpoints of a new edge is a high-degree vertex.
In particular, there are two cases to consider: 1) the bag already contains the high-degree vertex, and 2) the bag does not contain the high-degree vertex.
In both cases, the goal is to add the new vertex to one of the snowballs (if needed). 

Algorithm \ref{alg: vertex_addition_bag} defines the algorithm for vertex addition.
For the first case, the algorithm scans the bag to find the snowball that contains the vertex and returns it.
For the second case, the algorithm first identifies the candidate snowballs, then it assigns the vertex to one of the candidate snowballs.
The candidate snowballs are the ones having the main core number smaller than or equal to the internal degree of the new vertex (\intcore{\snowball}{u} $\geq$ \mcore{\snowball}).
Among the candidate snowballs, the new node is assigned to the snowball with maximum internal degree $\idegree{\snowball}{u}$, breaking ties randomly.

Once a vertex is added to a snowball, the core number of the snowball may increase. 
This change requires removing the vertices with core number lower than the main core of the snowball.
This procedure can be implemented efficiently in linear time by sorting the vertices based on their degree similar to bin sort.\footnote{The \textsc{MaintainInvariant} method at line \ref{fcall:mantaininvariant} of Algorithm \ref{alg: vertex_addition_bag}.}

\spara{Verification.}
Due to the greedy assignment of vertex to the snowball, it is possible that the vertex ends up not having the highest possible core number.
For example, Figure \ref{verify-example-1} shows an example where the greedy assignment does not result in optimal solution.

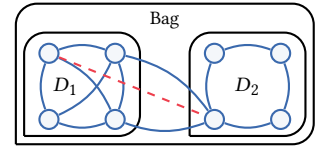
\begin{figure}[t]
	\centering
	\begin{tikzpicture}  [scale=1.1,every node/.style={scale=1.2}, shift=(current page.center)]]

\tikzstyle{vertex} = [font=\large, inner sep = 1pt,  minimum width=6pt, 
	thick, circle, text=black!90, draw=yafcolor5!90, fill=yafcolor5!05]
\tikzstyle{line} = [thick, yafcolor5!90]

\tikzset{
  text style/.style={text=black!70, font=\footnotesize}
}

\node[text width=4cm, font=\scriptsize, text centered] at (1.3,1.5) 
    {$D_1$};
\node[text width=4cm, font=\scriptsize, text centered] at (3.5,1.5) 
    {$D_2$};

\node[text width=4cm, font=\scriptsize, text centered] at (2.5,2.3) 
    {Bag};

\draw[thick,rounded corners=8pt] (0.8,0.9) -- (0.8,2.15) -- (2.2,2.15) -- (2.2,0.9) -- (0.8,0.9);
\draw[thick,rounded corners=8pt] (2.8,0.9) -- (2.8,2.15) -- (4.2,2.15) -- (4.2,0.9) -- (2.8,0.9);
\draw[thick,rounded corners=8pt] (0.7,0.8) -- (4.3,0.8) -- (4.3,2.5) -- (0.7,2.5) -- (0.7,0.8);

\node[vertex] (v1) at (1.1,1.1) {};
\node[vertex] (v2) at (1.9,1.1) {};
\node[vertex] (v3) at (1.1, 1.9) {};
\node[vertex] (v4) at (1.9,1.9) {};

\node[vertex] (v5) at (3.1,1.1) {};
\node[vertex] (v6) at (3.9,1.1) {};
\node[vertex] (v7) at (3.1,1.9) {};
\node[vertex] (v8) at (3.9,1.9) {};

\draw (v1) edge [line, -, bend right = 20] (v2);
\draw (v1) edge [line, -, bend left = 20] (v3);
\draw (v1) edge [line, -, bend right = 20] (v4);
\draw (v2) edge [line, -, bend right = 20] (v3);
\draw (v2) edge [line, -, bend right = 20] (v4);
\draw (v3) edge [line, -, bend left = 20] (v4);
\draw (v5) edge [line, -, bend right = 20] (v6);
\draw (v5) edge [line, -, bend left = 20] (v7);
\draw (v6) edge [line, -, bend right = 20] (v8);
\draw (v7) edge [line, -, bend left = 20] (v8);
\draw (v4) edge [line, -, bend left = 20] (v5);
\draw (v2) edge [line, -, bend right = 20] (v5);
\draw (v3) edge [line, dashed ,yafcolor4!90, -, bend left = 0] (v5);

\end{tikzpicture} 
\caption{Example showing that arrival of a new edge allows the vertex that is part of snowball $D_2$ to become part of snowball $D_1$, which has greater core number (3).} 
\label{verify-example-1}
\vspace{-\baselineskip}
\end{figure}
Therefore, after addition, the algorithm ensures that the core number of the snowball, where the new vertex is added, equals the core number of the new vertex in the graph.
The algorithm verifies that the core number of the added vertex by comparing it with the core number of the vertex in the bag.  
Note that the bag represents the supergraph containing all the snowballs and edges between the snowballs.
If the core number within the bag is larger than the one in the snowball,
the algorithm merges all the snowballs in the induced subgraph of newly added vertex. 
As all the vertices in the induced subgraph have the same core number, merging them ensures creating a larger snowball.\footnote{The \textsc{FixMainCore} method at line \ref{fcall:fixmaincore} of Algorithm \ref{alg: edge_addition}.}
We leverage the core decomposition algorithm by \citet{sariyuce2013streaming} for the implementation. 


\begin{algorithm}[tb]
\footnotesize
\caption{Vertex Addition in the Bag of Snowballs}
\label{alg: vertex_addition_bag}
\begin{algorithmic}[1]
\Procedure{addtoBag}{u}
	\State $\densest \leftarrow \emptyset$
	\For{\texttt{$\isnowball{i} \in \bag$}} \label{alg:line:loopsnow}
 		\If {$u \in \isnowball{i}$}
 			\State \Return $\isnowball{i}$ \Comment {First Case}
 		\EndIf
	 	\If{ ($\idegree{u}{\isnowball{i}} \geq \mcore{\isnowball{i}}$ and $\density{\densest} < \density{\isnowball{i}} $)}
			\State $\densest \leftarrow \isnowball{i}$	
		\EndIf 
    \EndFor
	\If {$\densest = \emptyset$}      
     	\State $\densest \gets \{u\}$
     \Else 
     	\State $\densest \gets \densest \cup \{u\}$
     	 \State \label{fcall:mantaininvariant} \Call{MaintainInvariant}{$\densest$} 
     \EndIf 
    	\State \Return $\densest$  \Comment {Second Case} 
\EndProcedure
\end{algorithmic}
\end{algorithm}

\begin{algorithm}[tb]
\footnotesize
\caption{Maintain Invariant}
\label{alg: maintain_invariant}
\begin{algorithmic}[1]
\Procedure{MaintainInvariant}{\isnowball{i}}
\Do 
	\State $repeat \gets false$
	\For{\texttt{$u \in \isnowball{i}$}}
		 \If{ ($(\intcore{u}{\isnowball{i}} < \mcore{\isnowball{i}} $) }
		 	\State $\isnowball{i} \gets \isnowball{i} \backslash \{u\}$
		 	\If {($\degree{u} \geq \estimatedensity $)}
		 		\State \Call{addtoBag}{u}
		 	\EndIf
		 	\State $repeat \gets true$
		 \EndIf
      \EndFor
\doWhile{$repeat$}
\EndProcedure
\end{algorithmic}
\end{algorithm}


\begin{algorithm}
\footnotesize
\caption{Fix Main Core}
\label{alg: fixMainCore}
\begin{algorithmic}[1]
\Procedure{fixMainCore}{}
\For {\texttt{$\isnowball{i} \in \bag$}} 
	\If {($\isnowball{i} \cap \induced{u} > $  0)}
		\State $\isnowball{u} \gets \isnowball{u} \cup \isnowball{i}$
	\EndIf
\EndFor
\State  \Call{MaintainInvariant}{\isnowball{u}}
\EndProcedure
\end{algorithmic}
\end{algorithm}

\begin{theorem}
Given the bag \bag, the algorithm ensures that \bag contains the main core of the graph within one of the snowballs after the vertex addition.
\end{theorem}
\begin{proof}
Let us assume that at time $t$ the bag contains the main core of the graph.
Now, we need to show that at time $t+1$, after the node addition, the bag still contains the main core of the graph.
In general, vertex addition method is called whenever there is an edge addition.
The only way for the new vertex to affect the main core of the graph is that the new vertex is the part of the main core.
After the addition of the vertex in the bag, the algorithm verifies the core number by comparing the core number of the vertex in the bag and the snowball.
If the core number of the vertex in the bag is greater, the algorithm merges the snowballs containing the vertices in the induced graph of the new node in the bag.
This creates a new snowball with a greater core number. 
\end{proof}

\spara{Edge addition:}
In this case, the state of the bag is only affected if at least one of the vertices in the new edge is a high-degree vertex.
In particular, there are two cases to consider: a) only one of the vertices is a high-degree vertex and b) both the vertices are high-degree vertices.
For the first case, the algorithm leverages the vertex addition method to add the vertex to the bag of snowballs.
For the second case, when both vertices are added to the bag of snowballs, the algorithm verifies that the main core exists in the bag. 
When both vertices are added to the same snowball, the algorithm adds the new edge to the same snowball and ensures that the invariant holds.
Conversely, when the two vertices are added to two different snowballs, the algorithm verifies if the vertices exist in each others' induced subgraphs and fixes the main core for both the vertices. 
Algorithm \ref{alg: edge_addition} describes the algorithm for edge addition.

\begin{algorithm}[tb]
\footnotesize
\caption{Edge Addition}
\label{alg: edge_addition}
\begin{algorithmic}[1]
\Procedure{addEdge}{(u,v)}
\If {(($\degree{u} < \estimatedensity) $) and ($ \degree{v} < \estimatedensity $))}
	\State \Return 
\ElsIf {(($\degree{u} \geq \estimatedensity $) and ($\degree{v} < \estimatedensity $))}
	\State $\isnowball{u} \gets$ \Call {addtoBag}{u}
\ElsIf {(($\degree{u} < \estimatedensity$) and ($\degree{v} \geq \estimatedensity $)}
	\State $\isnowball{v} \gets$ \Call {addtoBag}{v}
\Else
	\State $\isnowball{u} \gets$ \Call {addtoBag}{u}
	\State $\isnowball{v} \gets$ \Call {addtoBag}{v}
	\If {($\isnowball{u} = \isnowball{v}$)}
		\State $\isnowball{u} \gets \isnowball{u} \cup (u,v)$
		\State \Call{MaintainInvariant}{$\isnowball{u}$}
	\ElsIf {($v \in \induced{u}$)}
		\State \label{fcall:fixmaincore}\Call{fixMainCore}{u} 
	\EndIf
\EndIf
\EndProcedure
\end{algorithmic}
\end{algorithm}

\begin{theorem}
Algorithm \ref{alg: edge_addition} maintains the main core of the graph in one of the snowballs inside the bag.
\end{theorem}
\begin{proof}
The proof for all the cases, other than the case when both the vertices of the new edge are assigned to two different snowballs, is similar to the vertex addition algorithm.
Therefore, we consider the case when both end vertices of the added edge are added to two different snowballs.
For this case, we rely on the graph in the bag. 
We check if both vertices are in the same core graph in the bag, and fix the core number of the vertices if they belong to the same induced subgraph. 
This ensures creating the graph with the highest core number.
\end{proof}
\subsection{Removal operations}

\spara{Vertex removal:}
Similarly to the addition operations, we first define the procedure for removing a vertex from the bag of snowballs.
The vertex removal method is used as a subroutine for the edge removal process.
A vertex is only removed from the bag when its degree becomes lower than the maximal density. 
Therefore, according to Lemma~\ref{lema:2-approximation}, the removed vertex cannot be part of the main core.
The algorithm removes the vertex from the snowball within a bag without doing any other operation.



\begin{algorithm}[tb]
\footnotesize
\caption{Edge Deletion}
\label{alg: edge-deletion}
\begin{algorithmic}[1]
\Procedure{removeEdge}{(u,v)}
\If {(($\degree{u} <\estimatedensity $) and ($ \degree{v} < \estimatedensity $))}
	\State \Return
\EndIf
\If {(($\degree{u} < \estimatedensity$) and ($\degree{u}+1 \geq \estimatedensity $))}
	\State \Call {removeVertex}{u}
	\State \Return
\EndIf
\If {(($\degree{v} < \estimatedensity $) and ($\degree{v}+1 \geq \estimatedensity $))}
	\State \Call {removeVertex}{v}
	\State \Return
\EndIf
\For{\texttt{$\isnowball{i} \in \bag$}}
	 \If{($\isnowball{i} \cap (u,v) \neq \emptyset $)}
	 	\State $\isnowball{i} \gets \isnowball{i} \backslash (u,v)$
	 	\For{\texttt{$x \in \isnowball{i}$}}
	 		\If { ($\intcore{x}{B} > \intcore{x}{\isnowball{i}}$)} 
				\State \Call{fixMainCore}{x}
			\EndIf
		\EndFor
		\State \Call{MaintainInvariant}{$\isnowball{i}$}
	\EndIf
\EndFor
\EndProcedure
\end{algorithmic}
\end{algorithm}

\spara{Edge removal:}
Now we turn our attention to edge deletion, which follows the same pattern as edge addition.
Again, we leverage the bag to ensure that there exist a snowball with a core number equal to the main core of the graph.
Algorithm \ref{alg: edge-deletion} shows the algorithm for edge deletion.
The bag does not require any update when either one of the vertices is low-degree or if both the vertices belong to two different snowballs.
Therefore, we consider the case when one of the vertices lie at the boundary of high-degree vertices.
That is, the edge deletion moves the vertex from the high-degree to low-degree.
In this case, the algorithm only requires removing the vertex from the bag, without performing any other operations.

The interesting case is where both vertices of the deleted edge are high-degree and belong to the same snowball. 
In this case, the algorithm removes the edge from the snowball.
Further, it verifies and updates (if needed) the core number of the vertices affected by the update in the snowball.
Lastly, edge deletion might reduce the maximal density, and thus
require adding to the bag new vertices whose degree is now greater than the new maximal density.


\begin{theorem}
Given the bag containing a snowball with the same core number as the main core of the bag, after the edge deletion, Algorithm  \ref{alg: edge-deletion} maintains the main core of the graph in one of the snowballs inside the bag.
\end{theorem}
\begin{proof}
An edge removal affects the bag of snowballs only when both the vertices corresponding to the removed edge belong to the same snowball.
In this case, edge removal might reduce the core number of all the vertices in the snowball.
The algorithm ensures that the vertices in the snowball have their maximum possible core number by comparing their core number in the snowball with the core number in the bag.
Therefore, by verifying and fixing the core numbers, the algorithm maintains the main core of the graph in one of the snowballs inside the bag.
\end{proof}

\subsection{Fully-dynamic top-$k$ densest subgraphs}
Now that we have a fully dynamic algorithm for finding the densest subgraph in sliding windows, we move our attention to the top-$k$ densest-subgraph problem.

Let $\estimatedensity \geq \density{1} \geq \ldots \geq \density{z-1}$ represent the densities of the $z$ subgraphs in the bag, where $\density{\densest}$ is the density of the densest subgraph.
The bag contains the vertices that have a degree greater than the density $\estimatedensity$.
To ensure that the bag contains at least $k$ subgraphs, we modify the algorithm to keep all the vertices with degree greater than $\density{k-1}$.
The only modification required is to replace $\density{\densest}$ by  $\density{k-1}$ in Algorithm \ref{alg: maintain_invariant}, Algorithm \ref{alg: edge_addition} and Algorithm \ref{alg: edge-deletion}. 
Further, we leverage the priority queue to extract $\density{k-1}$ from the bag of snowballs.
This simple modification ensures that the bag contains at least $k$ subgraphs and enables accessing the top-$k$ densest subgraphs in the sliding window model.
Note that the new definition of high-degree vertices is related to the density of the $k$-th top densest subgraph. 

The modified algorithm guarantees that the bag contains the vertices with a degree greater than the $\density{k-1}$ after any graph update.
These graph updates include both edge additions and removals.
It is necessary for the algorithm to consider edge additions, as they might affect the value of $\density{k-1}$ by merging multiple subgraphs.
Similarly, edge deletions have to be considered, as they might affect the value of $\density{k-1}$ by reducing the density of any of the top-$k$ densest subgraphs, merging multiple subgraphs, or splitting a subgraph. 

As the algorithm ensures keeping the main core in the bag, it guarantees $2$-approximation for the first densest subgraph ($k=1$) while providing a high-quality heuristic for $k > 1$ compared to other solutions.
We provide an example in Figure \ref{fig:sample-figure-top-k} for the top-k densest subgraph algorithm by expanding on the Figure \ref{fig:sample-graph}.
We conclude this section with the theorem that provides a bound for our proposed algorithm.
These bounds can be generalized for any algorithm that adapts to a solution for densest subgraph problem for the top-$k$ case (see section \ref{subsection: topk} for examples).
\begin{theorem}
\label{theorem:bounds}
Given an integer $k$, the bag contains a set of subgraphs that provides $2k$-approximation of the top-$k$ densest-subgraph problem.
\end{theorem}
\begin{proof}
We know that the graph does not contain any subgraph with density greater than the optimal density ($\density{\densest}$).
Therefore, we know that the sum of densities for the top-$k$ densest-subgraph problem is upper bounded by $k\times \density{\densest}$.

Further, the bag contains \estimatedensity, which provides a 2-approximation of the densest subgraph.
This implies that the sum of densities of top-$k$ densest subgraph in the bag $\geq \frac{1}{2}\density{\densest}$.
Putting above two observations together, we can clearly see that the bag provides a solution that is a $2k$-approximation for the problem.
\end{proof}

\begin{figure}[t]
\begin{subfigure}{0.48\columnwidth}
\centering 
  \begin{tikzpicture} 

\tikzstyle{vertex} = [font=\small, inner sep = 1pt,  minimum width=6pt, 
	thick, circle, text=black!90, draw=yafcolor5!90, fill=yafcolor5!05]
\tikzstyle{line} = [thick, yafcolor5!90]

\tikzset{
  text style/.style={text=black!70, font=\footnotesize}
}

\node[vertex] (v1) at (1.1,1) {};
\node[vertex] (v2) at (1.55,1.55) {};
\node[vertex] (v3) at (1.35,2.15) {};
\node[vertex] (v4) at (0.65,2.1) {};
\node[vertex] (v5) at (0.5,1.5) {};

\node[vertex] (v6) at (0.0,1.1) {};
\node[vertex] (v7) at (-0.1,1.8) {};
\node[vertex] (v8) at (0.1,2.5) {};

\node[vertex] (v9) at (2.7,0.75) {};
\node[vertex] (v10) at (2.65,1.4) {};
\node[vertex] (v11) at (2.1,1.45) {};
\node[vertex] (v12) at (2.05,0.8) {};

\node[vertex] (v13) at (2.9,2.15) {};
\node[vertex] (v14) at (2.7,2.8) {};
\node[vertex] (v15) at (2.15,2.75) {};
\node[vertex] (v16) at (2.1,2.1) {};

\node[vertex] (v17) at (1.55,2.9) {};
\node[vertex] (v18) at (2.2,3.4) {};

\draw (v1) edge [line, -, bend right = 20] (v2);
\draw (v1) edge [line, -, bend right = 20] (v4);
\draw (v1) edge [line, -, bend left = 20] (v5);
\draw (v2) edge [line, -, bend right = 20] (v3);
\draw (v2) edge [line, -, bend right = 20] (v4);
\draw (v2) edge [line, -, bend right = 20] (v5);
\draw (v3) edge [line, -, bend right = 20] (v4);
\draw (v3) edge [line, -, bend right = 20] (v5);
\draw (v4) edge [line, -, bend right = 20] (v5);

\draw (v5) edge [line, -, bend left = 20] (v6);
\draw (v5) edge [line, -, bend right = 20] (v7);
\draw (v6) edge [line, -, bend left = 20] (v7);

\draw (v4) edge [line, -, bend right = 20] (v8);

\draw (v9) edge [line, -, bend right = 20] (v10);
\draw (v9) edge [line, yafcolor4!90, -, bend left = 0] (v11);
\draw (v9) edge [line, -, bend left = 20] (v12);
\draw (v10) edge [line, -, bend right = 20] (v11);
\draw (v10) edge [line, -, bend left = 20] (v12);
\draw (v11) edge [line, -, bend right = 20] (v12);

\draw (v2) edge [line, -, bend right = 20] (v11);
\draw (v11) edge [line, -, bend right = 20] (v16);

\draw (v13) edge [line, -, bend right = 20] (v14);
\draw (v13) edge [line, -, bend left = 20] (v15);
\draw (v13) edge [line, -, bend left = 20] (v16);
\draw (v14) edge [line, -, bend right = 20] (v15);
\draw (v14) edge [line, -, bend right = 20] (v16);
\draw (v15) edge [line, -, bend right = 20] (v16);

\draw (v15) edge [line, -, bend right = 20] (v17);
\draw (v15) edge [line, -, bend left = 20] (v18);

\end{tikzpicture}
  \caption{}
  \label{fig:sample-top-ka}
\end{subfigure}%
\begin{subfigure}{0.48\columnwidth}
  \centering
  \begin{tikzpicture} 

\tikzstyle{vertex} = [font=\small, inner sep = 1pt,  minimum width=6pt, 
	thick, circle, text=black!90, draw=yafcolor5!90, fill=yafcolor5!05]
\tikzstyle{line} = [thick, yafcolor5!90]

\tikzset{
  text style/.style={text=black!70, font=\footnotesize}
}

\node[vertex] (v1) at (1.1,1) {};
\node[vertex] (v2) at (1.55,1.55) {};
\node[vertex] (v3) at (1.35,2.15) {};
\node[vertex] (v4) at (0.65,2.1) {};
\node[vertex] (v5) at (0.5,1.5) {};

\node[vertex] (v6) at (0.0,1.1) {};
\node[vertex] (v7) at (-0.1,1.8) {};

\node[vertex, draw=yafcolor4!90, fill=yafcolor4!05] (v9) at (2.7,0.75) {};
\node[vertex] (v10) at (2.65,1.4) {};
\node[vertex, draw=yafcolor4!90, fill=yafcolor4!05] (v11) at (2.1,1.45) {};
\node[vertex] (v12) at (2.05,0.8) {};

\node[vertex] (v13) at (2.9,2.15) {};
\node[vertex] (v14) at (2.7,2.8) {};
\node[vertex] (v15) at (2.15,2.75) {};
\node[vertex] (v16) at (2.1,2.1) {};


\draw (v1) edge [line, -, bend right = 20] (v2);
\draw (v1) edge [line, -, bend right = 20] (v4);
\draw (v1) edge [line, -, bend left = 20] (v5);
\draw (v2) edge [line, -, bend right = 20] (v3);
\draw (v2) edge [line, -, bend right = 20] (v4);
\draw (v2) edge [line, -, bend right = 20] (v5);
\draw (v3) edge [line, -, bend right = 20] (v4);
\draw (v3) edge [line, -, bend right = 20] (v5);
\draw (v4) edge [line, -, bend right = 20] (v5);

\draw (v6) edge [line, -, bend left = 20] (v7);


\draw (v9) edge [line, -, bend right = 20] (v10);
\draw (v9) edge [line, yafcolor4!90, -, bend left = 0] (v11);
\draw (v9) edge [line, -, bend left = 20] (v12);
\draw (v10) edge [line, -, bend right = 20] (v11);
\draw (v10) edge [line, -, bend left = 20] (v12);
\draw (v11) edge [line, -, bend right = 20] (v12);


\draw (v13) edge [line, -, bend right = 20] (v14);
\draw (v13) edge [line, -, bend left = 20] (v15);
\draw (v13) edge [line, -, bend left = 20] (v16);
\draw (v14) edge [line, -, bend right = 20] (v15);
\draw (v14) edge [line, -, bend right = 20] (v16);
\draw (v15) edge [line, -, bend right = 20] (v16);


\end{tikzpicture}
  \caption{}
  \label{fig:sample-top-kb}
  	\vspace{-\baselineskip}
\end{subfigure}
\caption{
Example showing top-$k$ version of the algorithm.
The densities of top-3 subgraphs, before the arrival of the new edge in the bag are $1.8$, $1.5$ and $1.25$.
The algorithm stores all the vertices in the bag with degree greater than equal to $1.25$, as shown in Figure \ref{fig:sample-top-kb}.
After the arrival of the new edge, the update only affects one of the subgraphs in the bag.}
\label{fig:sample-figure-top-k}
\vspace{-\baselineskip}
\end{figure}
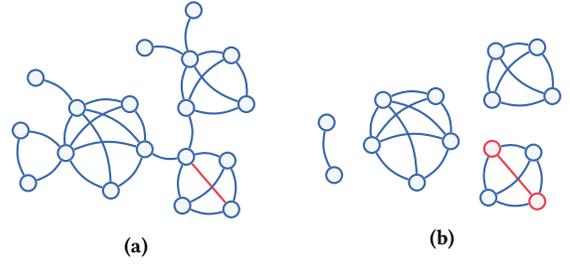
\subsection{Data structure}
\label{alg: data_structure}
The proposed algorithm requires accessing the neighborhood information of every node. 
Specifically, we are interested in performing three queries on a given vertex: a) extract its degree, b) extract all its neighbors, and c) given density $\estimatedensity$, extract all the vertices with degrees greater than the density $\estimatedensity$.

\spara{Vertex map:} To answer the first two queries, we need to store the neighborhood information for all vertices.
We store the information in a hashmap with keys being the vertex identifiers and values being the neighbors of each vertex.
Vertex map allows performing search and update operations in amortized constant time. 

\spara{Degree table:} For the third query, we need to order the vertices by their degrees.
We use bin sort to order the vertices by their degrees, which enables extracting the vertices with degrees greater than the density $\estimatedensity$ in constant time.
\section{Discussion}\label{sec:discussion}

Most of the solutions that leverage a static algorithms~\citep{charikar2000greedy,batagelj2003,bahmani2012densest} require iterating over the entire graph k times upon any update.
Comparatively, our algorithm mostly touches a limited region in the graph for any updates, which makes our algorithm perform significantly faster. 

The top-k densest subgraph algorithm that adapts to \citet{epasto2015efficient} requires replicating the graph across the $k$ instances.
In addition, updates across k instances might cascade and require running the algorithm k times, similarly to the static top-k solutions.
Our algorithm outperforms the pipeline version of the top-k algorithm both in terms of memory and computation. 

Finally, incremental algorithms for k-core maintenance requires removing top-k densest subgraph upon every update in the graph.
Each removal of a densest subgraph requires updating the core number of the remaining vertices.
Our algorithm efficiently avoids this removal step by maintaining the subgraphs during execution.

\spara{Performance.}
The proposed algorithm leverages the skew in real-world graphs and ignores a major portion of the input stream on the fly.
In addition, high-degree vertices are stored as small subgraphs so that each update is often applied  locally on these subgraphs.
This design allows our algorithm to operate on just small portions of the graph for each update, rather than iterating on the whole graph.
Finally, we use the k-core algorithm~\cite{sariyuce2013streaming,li2014efficient}, which allows each high-degree node to update their core number with a complexity independent of the graph size.
These improvements enable our algorithm to perform significantly better than the other state-of-the-art streaming algorithms.

\spara{Memory.} Our algorithm requires only $O(n^2 \polylog n)$ memory, compared to $O(k n^2 \polylog n)$ for the top-k version of \cite{epasto2015efficient} and $O(n^2 \polylog n)$ for the other algorithms discussed in Section~\ref{subsection: topk}. 

\spara{Tight Bounds.} We show via an example that any algorithm for densest subgraph problem can only produce a $k$-approximation for top-k densest subgraph problem. 
Consider a graph $G$ that contains $n\cdot m$ vertices, for any $n \geq k$.
The $n$ nodes connect to form a circle.
In the circle, there is one edge connecting two non-adjacent vertices. 
Additionally, each of the $n$ nodes connects to exactly $m$ neighbors.
The inner circle of the graph $G$ is the densest subgraph.
Removing the densest subgraph leaves the rest of the vertices completely disconnected.
Hence, the sum of density will be equal to one.
Now, consider the case when the densest subgraph is not removed first. 
In this case, each node in the circle along with its $m$ neighbors create a subgraph of density almost equal to one, for higher values of $m$.
Therefore, one can return $k$ such subgraphs with the sum of densities equal to $k$, which is $k$ times better than the previous case.

\section{Evaluation}
\label{sec:evaluation}
We conduct an extensive empirical evaluation of the proposed algorithm, and provide comparisons with the existing solutions.
In particular, we answer the following questions:
\begin{squishlist}
\item[\textbf{Q1:}]
What is the impact on the quality of subgraphs?
\item[\textbf{Q2:}]
What are the gains in performance?
\item[\textbf{Q3:}]
How does the algorithm perform in terms of different input parameters?
\end{squishlist}
\subsection{Experimental setup}
\spara{Datasets.}
Table~\ref{tab:summary-datasets} shows the datasets used in the experiments.
The datasets are selected due to their public availability. 
We evaluate all the algorithms in the sliding window model.
The number of edges in the sliding window is an input parameter.

\begin{table}[t]
\centering
\caption{Datasets used in the experiments.} 
\small
\begin{tabular}{l l r r r}
\toprule
Dataset		&	Symbol	&	$n$	&	$m$ & $\overline{\degree{v}}$	\\ 
\midrule
Amazon \cite{snapnets}		& AM		& \num{334863}	& \num{925872} & 5.52 \\
DBLP 1 \cite{snapnets}		&  DB$_1$   & \num{317080} &   \num{1049866} & 6.62 \\
Youtube \cite{snapnets} & YT & \num{1134890} & \num{2987624} & 5.26 \\
DBLP 2\tablefootnote{\url{http://konect.uni-koblenz.de/networks/dblp_coauthor}}			&	DB$_2$	& \num{1314050}  & 	\num{18986618} & 28.88  \\
Live Journal\tablefootnote{\url{http://konect.uni-koblenz.de/networks/livejournal-links}}  & LJ		&	\num{5204176} & \num{49174620}  & 18.90 \\
Orkut \cite{snapnets} & OT & \num{3072441} & \num{117185083} & 76.28 \\
Friendster	\cite{snapnets}	&	FR	&		\num{65608366} & \num{1806067135} & 55.04 \\
\bottomrule
\end{tabular}
\label{tab:summary-datasets}
\vspace{-\baselineskip}
\end{table}

\spara{Metrics.}
We evaluate the quality and efficiency of the algorithms.
We assess the quality by the objective function, i.e., the sum of densities of the subgraphs produced by an algorithm.
We evaluate the efficiency of an algorithm by reporting the average update time and the memory usage.
The average update time is the average time it takes to move the sliding window.
This includes adding the new edge, removing the oldest edge, and updating the top-$k$ densest subgraphs.
We report the memory usage as the average percentage of occupied memory.

\spara{Algorithms.}
Table \ref{tab:summary-algorithm} shows the notations used for different algorithms. 
The algorithms by \citet{bahmani2012densest} and \citet{epasto2015efficient} require an additional epsilon parameter for execution, which provides a trade-off between quality and execution time. 
Here we use the defaults proposed by the authors.


\spara{Stream ordering.} 
We consider two commonly used stream ordering schemes~\cite{stanton2012streaming}: 
\begin{squishlist}
\item{BFS:} The ordering is a result of a breadth-first search starting from a random vertex. 
\item{DFS:} The ordering is a result of a depth-first search starting from a random vertex.
\end{squishlist}

\spara{Experimental environment.}
We conduct our experiments on a machine with 2 Intel Xeon Processors E5-2698 and 128GiB of memory. 
All the algorithms are implemented in Java and executed on JRE 7 running on Linux.
The source code is available online.\footnote{\url{https://github.com/anisnasir/TopKDensestSubgraph}}

\begin{table}[t]
\caption{Notation for the top-$k$ algorithms.}
\centering
\small
\begin{tabular}{l l r r}
\toprule
Symbol & Reference Algorithm	&	Top-1 Approx. & In-place \\
& & Guarantees & \\ 
\midrule
\charikar & \citet{charikar2000greedy}	& 2 & No \\
\batagelj & \citet{batagelj2003}	& 2 &  No \\
\bahmani{\epsilon} & \citet{bahmani2012densest} & $2(1+\epsilon )$ &  No \\
\li & \citet{li2014efficient} & 2  & Yes \\ 
\sariyuce & \citet{sariyuce2013streaming} & 2  &  Yes \\ 
\epasto{\epsilon} & \citet{epasto2015efficient}& $2(1+\epsilon )^6$ &  No \\
\our & this paper	& 2 & Yes\\
\bottomrule
\end{tabular}
\label{tab:summary-algorithm}
\vspace{-\baselineskip}
\end{table}

\subsection{Experimental results}

\spara{Q1:} In this experiment, we compare the quality of the results produced by the different algorithm as measured by the objective function.
In order to be able to run most of the algorithms, we use the smaller datasets, i.e., AM, DB$_1$, YT, and DB$_2$. 
As \li and \sariyuce produce same results in terms of quality, we only report the results for one of them.
Due to space constraints, we show the results only with the DFS ordering.
However, we achieve similar results also with the BFS one.
We set the size of the sliding window $x = 100$k.

\begin{figure*}[t]
    \centering
  		\includegraphics[width=\textwidth]{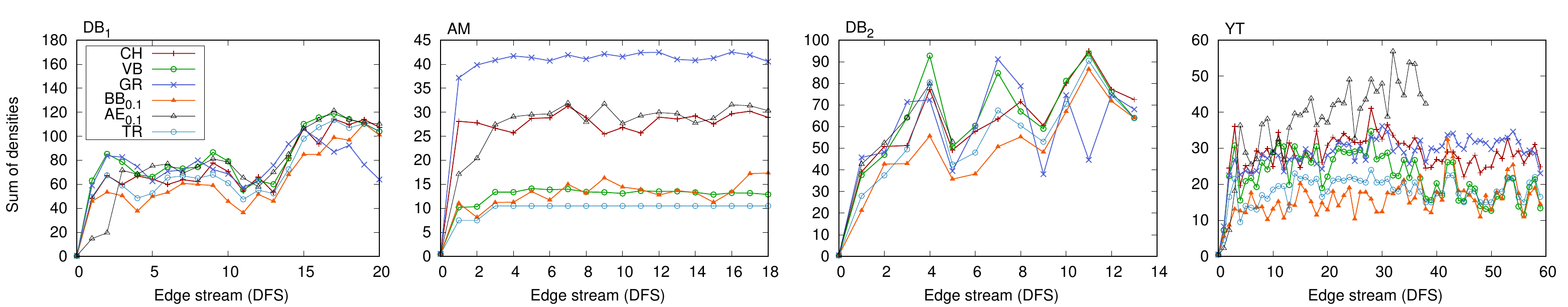}
 	\caption{Sum of densities for top-10 densest subgraphs on the DB$_1$, AM, DB$_2$ and YT datasets, sliding window size 100k.}
	\label{fig:eval-sum-densities}
	\vspace{-\baselineskip}
\end{figure*}

For the static algorithms, we execute them in micro-batches in which the top-$k$ densest subgraph is recomputed after $1$k edge updates, which gives them a substantial advantage.
For \bahmani{\epsilon} and \epasto{\epsilon}, we use $\epsilon = 0.01$. 
We set the maximum execution time to 7 days, due to which the \epasto{\epsilon} is not able to finish on DB$_2$ and YT.
Finally, \li and \sariyuce contain an update phase due to the iterative extraction of the top-$k$ densest subgraph. 
This phase is expensive, as the main core consists of high-degree vertices.
To alleviate this issue, we execute the extraction phase every $10$k update operations.

Figure \ref{fig:eval-sum-densities} shows the quality results.
Our algorithm \our achieves competitive quality, often generating denser subgraphs compared to all the other algorithms.
For example, for the AM dataset, \our produces subgraphs that are $1.5$ times better than the best algorithm among the state-of-the-art solutions.
All the other algorithms produce consistent results across all the datasets.
For instance, \batagelj , \li, and \sariyuce produce top-$k$ dense subgraphs of lowest quality compared to the other algorithms. 
This result is caused by the removal of the main core (backbone) of the graph in each of the $k$ iterations.
In addition, it shows how the problem considered in this paper, while related, is different from a simple k-core decomposition.
The results also validate that \bahmani{\epsilon} and \epasto{\epsilon} provide weaker guarantees on the quality compared to \charikar.

\spara{Q2:} In this experiment, we turn our attention to evaluate the efficiency of the proposed algorithm, as measured by the average update time and memory usage.
We again select several datasets, i.e.,  AM, DB$_1$, YT, DB$_2$, and LJ, and execute \charikar, \batagelj, \bahmani, \li , \sariyuce, \epasto, and \our.
We exclude the largest datasets as only a few algorithms are able to handle them.
The algorithms are executed with both BFS and DFS ordering of the edge stream.
The sliding window size is set to $x=100$k edges, and $k=10$ unless otherwise specified.

Figure \ref{fig:eval-performance} shows the efficiency results (note the logarithmic scale). 
The proposed algorithm, \our, outperforms all other ones in terms of update time for all the datasets and both ordering schemes.
\li and \sariyuce are the slowest algorithms.
In particular, for DFS, our algorithm achieves a performance gain of five orders of magnitude.
This result is noteworthy as even though our algorithm is dependent on core decomposition, it is still able to beat a na\"{i}ve application of those algorithms by a wide margin.
This difference is due to the fact that both algorithms require maintaining the core number for all the vertices.
Additionally, the extraction of top-$k$ densest subgraph further hampers their efficiency.

Algorithms 
\charikar, \batagelj, and \bahmani{0.01} perform unfavorably due to their static nature.
In this case, our algorithm is able to achieve more than three orders of magnitude improvement in efficiency.

Algorithm 
\epasto{0.01} is best performing among the baselines, 
and even outperforms \our in one dataset for one specific ordering (AM with BFS).
However, \our still outperforms \epasto{0.01} by nearly three orders of magnitude in most other cases.

Finally, we observe the overall memory consumption of all the algorithms (reported in Table \ref{tab:memory}). 
The memory requirement of our algorithm lies between the static and the dynamic algorithms, i.e., \epasto{0.01} requires the largest amount of memory, while the static algorithms require the least.
These results are in line with our expectations from the discussion in Section~\ref{sec:background}.

\begin{table}[tb]
\centering
\caption{Memory consumption as a percentage of total memory for algorithms with DB$_2$ and LJ dataset, sliding window size 100k.}
\small
\begin{tabular}{l c c c c c c c}
\toprule
		&	\charikar & \batagelj & \bahmani{\epsilon} &  \li  & \sariyuce & \epasto{\epsilon} & \our\\ 
\midrule
DB$_2$  &  2.10 & 2.10 & 2.10 & 2.1 & 2.8 & 5 & 2.7 \\
LJ & 1.80 & 1.80 & 1.80 &1.8 &2.1 & 5.5 & 2.4\\
\bottomrule
\end{tabular}
\label{tab:memory}
\end{table}

\begin{figure}[t]
    \centering
  		\includegraphics[width=1.05\columnwidth]{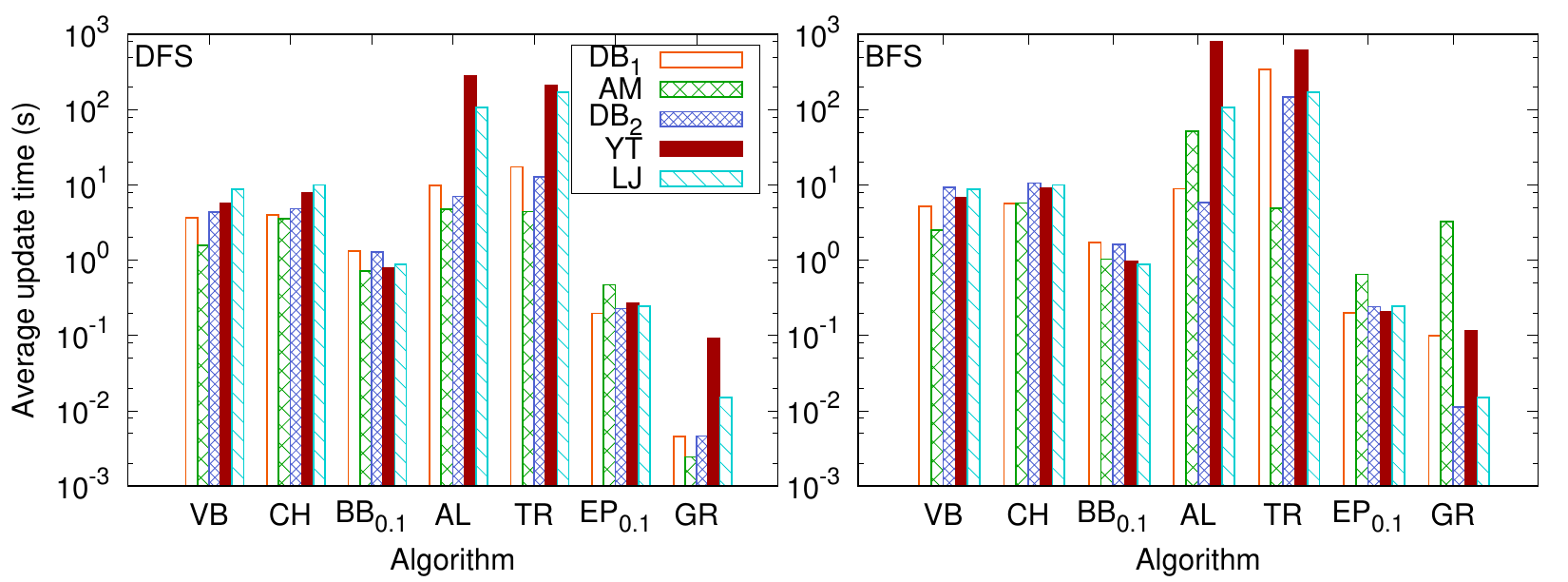}
 	\caption{Update time for the algorithms on the DB$_1$, AM, DB$_2$, YT, and LJ datasets, when using both DFS and BFS ordering and $k=10$.}
	\label{fig:eval-performance}
	\vspace{-\baselineskip}
\end{figure}

\spara{Q3:} In this experiment, we evaluate the scalability of our algorithm.
First, we execute \our with different values of $k$.
Alongside, we report the average update time for \bahmani{\epsilon} algorithm.
We choose \bahmani{\epsilon}, rather than \epasto{\epsilon}, as it can execute in mirco-batches, i.e., top-$k$ densest subgraph every $100$ edges.
We set the sliding window size $x=1$M, and use the YT and DB$_2$ dataset with DFS ordering.
Figure \ref{fig:eval-scal-1} reports the average update time for both algorithms. 
The average update time of \our remains consistent even for higher values of $k$, whereas the execution time of \bahmani{0.01} increases with the parameter.



\begin{figure}[t]
    \centering
  		\includegraphics[width=\columnwidth]{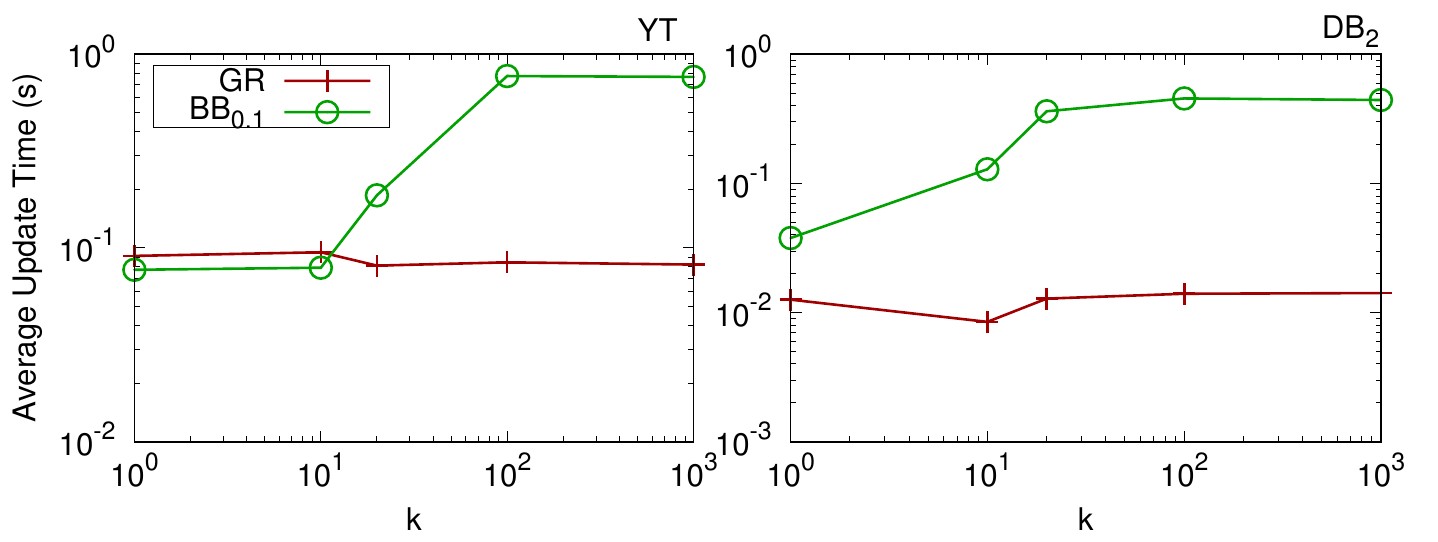}
 	\caption{Update time for \our and \bahmani{0.01} on the YT and DB$_2$ datasets as a function of $k$. Note that \bahmani{0.01} is executed in micro-batches of size $100$.}
	\label{fig:eval-scal-1}
	\vspace{-\baselineskip}
\end{figure}


Second, we execute \our with different sizes of sliding window, i.e., $x=10$k, $100$k, $1$M, and $10$M. We set $k=10$ for this experiment.
Again, we use the YT and DB$_2$ datasets with DFS ordering. 
Table~\ref{tab:scal-2} reports the average update time for different configurations. 
Increasing the size of the sliding window does not affect the average update time significantly.
This result validates our claim that most of the updates are local to the subgraphs, and do not require iterating through the whole graph to extract the top-$k$ densest subgraphs.

\begin{table}[tb]
\centering
\caption{Average update time for the \our algorithm with sliding windows of different sizes $x$.}
\small
\begin{tabular}{l c c c c}
\toprule
		&	$x=10$k 	&	$x=100$k	&	$x=1$M & $x=10$M	\\ 
\midrule
YT & 0.80ms & 91.14ms & 90.33ms & 95.49ms\\
DB$_2$  & 4.97ms & 7.58ms & 8.45ms & 32.21ms \\ 
\bottomrule
\end{tabular}
\label{tab:scal-2}
\vspace{-\baselineskip}
\end{table}


Finally, we study the performance of our algorithm in on the largest datasets.
We select the OT and FR datasets with DFS ordering, and execute the algorithm with a sliding window of $x=100$k, with $k=10$.
Figure \ref{fig:eval-scal-3} shows the result of the experiment. 
The plot is generated by taking the moving average of the update time and the sum of densities. 
The average update time of the algorithm mostly remains constant throughout the execution, and our algorithm provides steady efficiency over the stream. 
This behavior remains consistent even when the densities are fluctuating, as in the case for the OT dataset. 

\begin{figure}[t]
    \centering
  		\includegraphics[width=\columnwidth]{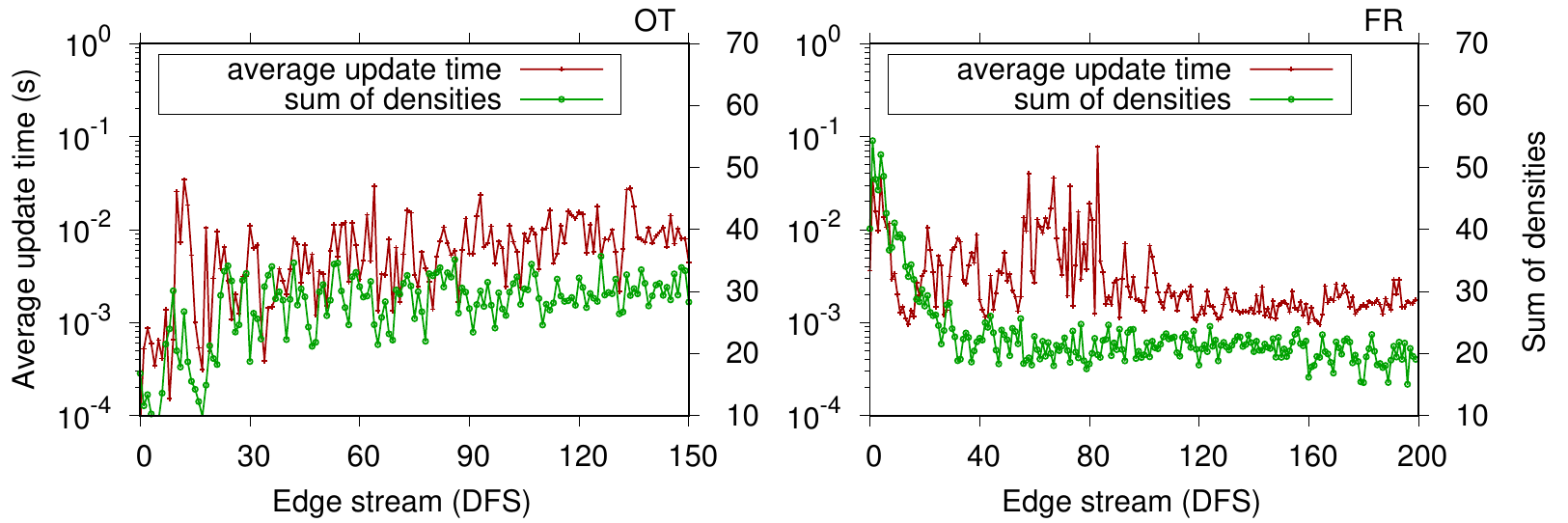}
 	\caption{Quality and efficiency for \our on the OT and FR datasets over the stream.}
	\label{fig:eval-scal-3}
	\vspace{-\baselineskip}
\end{figure}

\section{Related work}
\label{sec:rel-work}

\citet{valari2012discovery} were the first one to study the top-$k$ densest subgraph problem for a stream consisting of a dynamic collection of graphs.
They proposed both an exact and an approximation algorithm for top-$k$ densest subgraph discovery.
Similar to our algorithm, the proposed algorithm relies on the core decomposition to provide the approximation guarantees.
The top-$k$ densest subgraphs produced by the algorithm are edge-disjoint.
\citet{balalau2015finding} studied the problem of finding the top-$k$ densest subgraph with limited overlap.
They defined the top-$k$ densest subgraphs as a set of k subgraphs that maximizes the sum of densities, while satisfying an upper bound on the pairwise Jaccard coefficient between the set of vertices of the subgraphs.
The problem of finding the top-$k$ densest subgraph as sum of densities was shown to be NP-hard \cite{balalau2015finding} and efficient heuristic was proposed to solve the problem.
Further, \citet{galbrun2016top} studied the problem of finding the top-$k$ overlapping densest subgraphs and provided constant-factor approximation guarantees. 

\section{Conclusion}
We studied the top-$k$ densest subgraphs problem for graph streams, and proposed an efficient one-pass fully-dynamic algorithm.
In contrast to the existing state-of-the-art solutions that require iterating over the entire graph upon update, our algorithm maintains the solution in one-pass.
Additionally, the memory requirement of the algorithm is independent of $k$.
The algorithm is designed by leveraging the observation that graph updates only affect a limited region.
Therefore, the top-$k$ densest subgraphs are maintained by simply applying local updates to small subgraphs, rather than the complete graph.
We provided a theoretical analysis of the proposed algorithm and showed empirically that the algorithm often generates denser subgraphs than the state-of-the-art solutions. 
Further, we observed an improvement in performance of up to five orders of magnitude when compared to the baselines. 

This work gives rise to further interesting research questions:
Is it necessary to leverage k-core decomposition algorithm as a backbone?
Is it possible to achieve stronger bounds on the threshold for high-degree vertices?
Can we design an algorithm with a space bound on the size of the bag?
Is it possible to achieve stronger approximation guarantees for the problem?
We believe that solving these questions will further enhance the proposed algorithm, making it a useful tool for numerous practical applications.






\balance
\bibliographystyle{ACM-Reference-Format}
\bibliography{biblio}


\begin{thebibliography}{00}


\ifx \showCODEN    \undefined \def \showCODEN     #1{\unskip}     \fi
\ifx \showDOI      \undefined \def \showDOI       #1{{\tt DOI:}\penalty0{#1}\ }
  \fi
\ifx \showISBNx    \undefined \def \showISBNx     #1{\unskip}     \fi
\ifx \showISBNxiii \undefined \def \showISBNxiii  #1{\unskip}     \fi
\ifx \showISSN     \undefined \def \showISSN      #1{\unskip}     \fi
\ifx \showLCCN     \undefined \def \showLCCN      #1{\unskip}     \fi
\ifx \shownote     \undefined \def \shownote      #1{#1}          \fi
\ifx \showarticletitle \undefined \def \showarticletitle #1{#1}   \fi
\ifx \showURL      \undefined \def \showURL       #1{#1}          \fi
\providecommand\bibfield[2]{#2}
\providecommand\bibinfo[2]{#2}
\providecommand\natexlab[1]{#1}
\providecommand\showeprint[2][]{arXiv:#2}

\bibitem[\protect\citeauthoryear{Akiba, Iwata, and Yoshida}{Akiba
  et~al\mbox{.}}{2013}]%
        {akiba2013fast}
\bibfield{author}{\bibinfo{person}{Takuya Akiba}, \bibinfo{person}{Yoichi
  Iwata}, {and} \bibinfo{person}{Yuichi Yoshida}.}
  \bibinfo{year}{2013}\natexlab{}.
\newblock \showarticletitle{Fast exact shortest-path distance queries on large
  networks by pruned landmark labeling}. In \bibinfo{booktitle}{{\em SIGMOD}}.
  ACM, \bibinfo{pages}{349--360}.
\newblock


\bibitem[\protect\citeauthoryear{Angel, Sarkas, Koudas, and Srivastava}{Angel
  et~al\mbox{.}}{2012}]%
        {angel2012dense}
\bibfield{author}{\bibinfo{person}{Albert Angel}, \bibinfo{person}{Nikos
  Sarkas}, \bibinfo{person}{Nick Koudas}, {and} \bibinfo{person}{Divesh
  Srivastava}.} \bibinfo{year}{2012}\natexlab{}.
\newblock \showarticletitle{Dense subgraph maintenance under streaming edge
  weight updates for real-time story identification}.
\newblock \bibinfo{journal}{{\em VLDB\/}} \bibinfo{volume}{5},
  \bibinfo{number}{6} (\bibinfo{year}{2012}), \bibinfo{pages}{574--585}.
\newblock


\bibitem[\protect\citeauthoryear{Asahiro, Iwama, Tamaki, and Tokuyama}{Asahiro
  et~al\mbox{.}}{1996}]%
        {asahiro1996greedily}
\bibfield{author}{\bibinfo{person}{Yuichi Asahiro}, \bibinfo{person}{Kazuo
  Iwama}, \bibinfo{person}{Hisao Tamaki}, {and} \bibinfo{person}{Takeshi
  Tokuyama}.} \bibinfo{year}{1996}\natexlab{}.
\newblock \showarticletitle{{Greedily finding a dense subgraph}}. In
  \bibinfo{booktitle}{{\em SWAT}}. \bibinfo{pages}{136--148}.
\newblock


\bibitem[\protect\citeauthoryear{Babcock, Babu, Datar, Motwani, and
  Widom}{Babcock et~al\mbox{.}}{2002}]%
        {babcock2002models}
\bibfield{author}{\bibinfo{person}{Brian Babcock}, \bibinfo{person}{Shivnath
  Babu}, \bibinfo{person}{Mayur Datar}, \bibinfo{person}{Rajeev Motwani}, {and}
  \bibinfo{person}{Jennifer Widom}.} \bibinfo{year}{2002}\natexlab{}.
\newblock \showarticletitle{Models and issues in data stream systems}. In
  \bibinfo{booktitle}{{\em PODS}}. ACM, \bibinfo{pages}{1--16}.
\newblock


\bibitem[\protect\citeauthoryear{Bahmani, Kumar, and Vassilvitskii}{Bahmani
  et~al\mbox{.}}{2012}]%
        {bahmani2012densest}
\bibfield{author}{\bibinfo{person}{Bahman Bahmani}, \bibinfo{person}{Ravi
  Kumar}, {and} \bibinfo{person}{Sergei Vassilvitskii}.}
  \bibinfo{year}{2012}\natexlab{}.
\newblock \showarticletitle{Densest subgraph in streaming and mapreduce}.
\newblock \bibinfo{journal}{{\em VLDB\/}} \bibinfo{volume}{5},
  \bibinfo{number}{5} (\bibinfo{year}{2012}), \bibinfo{pages}{454--465}.
\newblock


\bibitem[\protect\citeauthoryear{Balalau, Bonchi, Chan, Gullo, and
  Sozio}{Balalau et~al\mbox{.}}{2015}]%
        {balalau2015finding}
\bibfield{author}{\bibinfo{person}{Oana~Denisa Balalau},
  \bibinfo{person}{Francesco Bonchi}, \bibinfo{person}{TH Chan},
  \bibinfo{person}{Francesco Gullo}, {and} \bibinfo{person}{Mauro Sozio}.}
  \bibinfo{year}{2015}\natexlab{}.
\newblock \showarticletitle{Finding subgraphs with maximum total density and
  limited overlap}. In \bibinfo{booktitle}{{\em WSDM}}. ACM,
  \bibinfo{pages}{379--388}.
\newblock


\bibitem[\protect\citeauthoryear{Batagelj and Zaversnik}{Batagelj and
  Zaversnik}{2003}]%
        {batagelj2003}
\bibfield{author}{\bibinfo{person}{Vladimir Batagelj} {and}
  \bibinfo{person}{Matjaz Zaversnik}.} \bibinfo{year}{2003}\natexlab{}.
\newblock \showarticletitle{An O (m) algorithm for cores decomposition of
  networks}.
\newblock \bibinfo{journal}{{\em arXiv preprint cs/0310049\/}}
  (\bibinfo{year}{2003}).
\newblock


\bibitem[\protect\citeauthoryear{Beutel, Xu, Guruswami, Palow, and
  Faloutsos}{Beutel et~al\mbox{.}}{2013}]%
        {beutel2013copycatch}
\bibfield{author}{\bibinfo{person}{Alex Beutel}, \bibinfo{person}{Wanhong Xu},
  \bibinfo{person}{Venkatesan Guruswami}, \bibinfo{person}{Christopher Palow},
  {and} \bibinfo{person}{Christos Faloutsos}.} \bibinfo{year}{2013}\natexlab{}.
\newblock \showarticletitle{Copycatch: stopping group attacks by spotting
  lockstep behavior in social networks}. In \bibinfo{booktitle}{{\em WWW}}.
  \bibinfo{pages}{119--130}.
\newblock


\bibitem[\protect\citeauthoryear{Bhattacharya, Henzinger, Nanongkai, and
  Tsourakakis}{Bhattacharya et~al\mbox{.}}{2015}]%
        {bhattacharya2015space}
\bibfield{author}{\bibinfo{person}{Sayan Bhattacharya}, \bibinfo{person}{Monika
  Henzinger}, \bibinfo{person}{Danupon Nanongkai}, {and}
  \bibinfo{person}{Charalampos Tsourakakis}.} \bibinfo{year}{2015}\natexlab{}.
\newblock \showarticletitle{Space-and time-efficient algorithm for maintaining
  dense subgraphs on one-pass dynamic streams}. In \bibinfo{booktitle}{{\em
  STOC}}. ACM, \bibinfo{pages}{173--182}.
\newblock


\bibitem[\protect\citeauthoryear{Charikar}{Charikar}{2000}]%
        {charikar2000greedy}
\bibfield{author}{\bibinfo{person}{Moses Charikar}.}
  \bibinfo{year}{2000}\natexlab{}.
\newblock \showarticletitle{Greedy approximation algorithms for finding dense
  components in a graph}. In \bibinfo{booktitle}{{\em Approx. Algo. for Comb.
  Opt.}} \bibinfo{pages}{84--95}.
\newblock


\bibitem[\protect\citeauthoryear{Chen and Saad}{Chen and Saad}{2012}]%
        {chen2012dense}
\bibfield{author}{\bibinfo{person}{Jie Chen} {and} \bibinfo{person}{Yousef
  Saad}.} \bibinfo{year}{2012}\natexlab{}.
\newblock \showarticletitle{Dense subgraph extraction with application to
  community detection}.
\newblock \bibinfo{journal}{{\em TKDE\/}} \bibinfo{volume}{24},
  \bibinfo{number}{7} (\bibinfo{year}{2012}), \bibinfo{pages}{1216--1230}.
\newblock


\bibitem[\protect\citeauthoryear{Crouch, McGregor, and Stubbs}{Crouch
  et~al\mbox{.}}{2013}]%
        {crouch2013dynamic}
\bibfield{author}{\bibinfo{person}{Michael~S Crouch}, \bibinfo{person}{Andrew
  McGregor}, {and} \bibinfo{person}{Daniel Stubbs}.}
  \bibinfo{year}{2013}\natexlab{}.
\newblock \showarticletitle{Dynamic graphs in the sliding-window model}. In
  \bibinfo{booktitle}{{\em ESA}}. Springer, \bibinfo{pages}{337--348}.
\newblock


\bibitem[\protect\citeauthoryear{Datar, Gionis, Indyk, and Motwani}{Datar
  et~al\mbox{.}}{2002}]%
        {datar2002maintaining}
\bibfield{author}{\bibinfo{person}{Mayur Datar}, \bibinfo{person}{Aristides
  Gionis}, \bibinfo{person}{Piotr Indyk}, {and} \bibinfo{person}{Rajeev
  Motwani}.} \bibinfo{year}{2002}\natexlab{}.
\newblock \showarticletitle{Maintaining stream statistics over sliding
  windows}.
\newblock \bibinfo{journal}{{\it SIAM J. Comput.}} \bibinfo{volume}{31},
  \bibinfo{number}{6} (\bibinfo{year}{2002}), \bibinfo{pages}{1794--1813}.
\newblock


\bibitem[\protect\citeauthoryear{Dourisboure, Geraci, and
  Pellegrini}{Dourisboure et~al\mbox{.}}{2007}]%
        {dourisboure2007extraction}
\bibfield{author}{\bibinfo{person}{Yon Dourisboure}, \bibinfo{person}{Filippo
  Geraci}, {and} \bibinfo{person}{Marco Pellegrini}.}
  \bibinfo{year}{2007}\natexlab{}.
\newblock \showarticletitle{Extraction and classification of dense communities
  in the web}. In \bibinfo{booktitle}{{\em WWW}}. ACM,
  \bibinfo{pages}{461--470}.
\newblock


\bibitem[\protect\citeauthoryear{Epasto, Lattanzi, and Sozio}{Epasto
  et~al\mbox{.}}{2015}]%
        {epasto2015efficient}
\bibfield{author}{\bibinfo{person}{Alessandro Epasto}, \bibinfo{person}{Silvio
  Lattanzi}, {and} \bibinfo{person}{Mauro Sozio}.}
  \bibinfo{year}{2015}\natexlab{}.
\newblock \showarticletitle{Efficient Densest Subgraph Computation in Evolving
  Graphs}. In \bibinfo{booktitle}{{\em WWW}}. \bibinfo{pages}{300--310}.
\newblock


\bibitem[\protect\citeauthoryear{Eppstein, Galil, and Italiano}{Eppstein
  et~al\mbox{.}}{1998}]%
        {eppstein1998dynamic}
\bibfield{author}{\bibinfo{person}{David Eppstein}, \bibinfo{person}{Zvi
  Galil}, {and} \bibinfo{person}{Giuseppe~F Italiano}.}
  \bibinfo{year}{1998}\natexlab{}.
\newblock \bibinfo{booktitle}{{\em Dynamic graph algorithms}}.
\newblock \bibinfo{publisher}{Springer}.
\newblock


\bibitem[\protect\citeauthoryear{Galbrun, Gionis, and Tatti}{Galbrun
  et~al\mbox{.}}{2016}]%
        {galbrun2016top}
\bibfield{author}{\bibinfo{person}{Esther Galbrun}, \bibinfo{person}{Aristides
  Gionis}, {and} \bibinfo{person}{Nikolaj Tatti}.}
  \bibinfo{year}{2016}\natexlab{}.
\newblock \showarticletitle{Top-$k$ overlapping densest subgraphs}.
\newblock \bibinfo{journal}{{\em Data Mining and Knowledge Discovery\/}}
  (\bibinfo{year}{2016}), \bibinfo{pages}{1--32}.
\newblock


\bibitem[\protect\citeauthoryear{Gibson, Kumar, and Tomkins}{Gibson
  et~al\mbox{.}}{2005}]%
        {gibson2005discovering}
\bibfield{author}{\bibinfo{person}{David Gibson}, \bibinfo{person}{Ravi Kumar},
  {and} \bibinfo{person}{Andrew Tomkins}.} \bibinfo{year}{2005}\natexlab{}.
\newblock \showarticletitle{Discovering large dense subgraphs in massive
  graphs}. In \bibinfo{booktitle}{{\em VLDB}}. VLDB Endowment,
  \bibinfo{pages}{721--732}.
\newblock


\bibitem[\protect\citeauthoryear{Gionis and Tsourakakis}{Gionis and
  Tsourakakis}{2015}]%
        {gionis2015dense}
\bibfield{author}{\bibinfo{person}{Aristides Gionis} {and}
  \bibinfo{person}{Charalampos~E Tsourakakis}.}
  \bibinfo{year}{2015}\natexlab{}.
\newblock \showarticletitle{Dense subgraph discovery: Kdd 2015 tutorial}. In
  \bibinfo{booktitle}{{\em SIGKDD}}. ACM, \bibinfo{pages}{2313--2314}.
\newblock


\bibitem[\protect\citeauthoryear{Goldberg}{Goldberg}{1984}]%
        {goldberg1984finding}
\bibfield{author}{\bibinfo{person}{Andrew~V Goldberg}.}
  \bibinfo{year}{1984}\natexlab{}.
\newblock \bibinfo{booktitle}{{\em Finding a maximum density subgraph}}.
\newblock \bibinfo{publisher}{University of California Berkeley, CA}.
\newblock


\bibitem[\protect\citeauthoryear{Kortsarz and Peleg}{Kortsarz and
  Peleg}{1994}]%
        {kortsarz1994generating}
\bibfield{author}{\bibinfo{person}{Guy Kortsarz} {and} \bibinfo{person}{David
  Peleg}.} \bibinfo{year}{1994}\natexlab{}.
\newblock \showarticletitle{Generating sparse 2-spanners}.
\newblock \bibinfo{journal}{{\em Journal of Algorithms\/}}
  \bibinfo{volume}{17}, \bibinfo{number}{2} (\bibinfo{year}{1994}),
  \bibinfo{pages}{222--236}.
\newblock


\bibitem[\protect\citeauthoryear{Leskovec and Krevl}{Leskovec and
  Krevl}{2014}]%
        {snapnets}
\bibfield{author}{\bibinfo{person}{Jure Leskovec} {and} \bibinfo{person}{Andrej
  Krevl}.} \bibinfo{year}{2014}\natexlab{}.
\newblock \bibinfo{title}{{SNAP Datasets}: {Stanford} Large Network Dataset
  Collection}.
\newblock \bibinfo{howpublished}{\url{http://snap.stanford.edu/data}}.
  (\bibinfo{date}{June} \bibinfo{year}{2014}).
\newblock


\bibitem[\protect\citeauthoryear{Li, Yu, and Mao}{Li et~al\mbox{.}}{2014}]%
        {li2014efficient}
\bibfield{author}{\bibinfo{person}{Rong-Hua Li}, \bibinfo{person}{Jeffrey~Xu
  Yu}, {and} \bibinfo{person}{Rui Mao}.} \bibinfo{year}{2014}\natexlab{}.
\newblock \showarticletitle{Efficient core maintenance in large dynamic
  graphs}.
\newblock \bibinfo{journal}{{\em TKDE\/}} \bibinfo{volume}{26},
  \bibinfo{number}{10} (\bibinfo{year}{2014}), \bibinfo{pages}{2453--2465}.
\newblock


\bibitem[\protect\citeauthoryear{McGregor, Tench, Vorotnikova, and Vu}{McGregor
  et~al\mbox{.}}{2015}]%
        {mcgregor2015densest}
\bibfield{author}{\bibinfo{person}{Andrew McGregor}, \bibinfo{person}{David
  Tench}, \bibinfo{person}{Sofya Vorotnikova}, {and} \bibinfo{person}{Hoa~T
  Vu}.} \bibinfo{year}{2015}\natexlab{}.
\newblock \showarticletitle{Densest Subgraph in Dynamic Graph Streams}.
\newblock In \bibinfo{booktitle}{{\em MFCS}}. \bibinfo{publisher}{Springer},
  \bibinfo{pages}{472--482}.
\newblock


\bibitem[\protect\citeauthoryear{Orlin}{Orlin}{2013}]%
        {orlin2013maxflows}
\bibfield{author}{\bibinfo{person}{James Orlin}.}
  \bibinfo{year}{2013}\natexlab{}.
\newblock \showarticletitle{{Max flows in $\bigO(nm)$ time, or better}}. In
  \bibinfo{booktitle}{{\em STOC}}. \bibinfo{pages}{765--774}.
\newblock


\bibitem[\protect\citeauthoryear{Rozenshtein, Anagnostopoulos, Gionis, and
  Tatti}{Rozenshtein et~al\mbox{.}}{2014}]%
        {rozenshtein2014event}
\bibfield{author}{\bibinfo{person}{Polina Rozenshtein}, \bibinfo{person}{Aris
  Anagnostopoulos}, \bibinfo{person}{Aristides Gionis}, {and}
  \bibinfo{person}{Nikolaj Tatti}.} \bibinfo{year}{2014}\natexlab{}.
\newblock \showarticletitle{Event detection in activity networks}. In
  \bibinfo{booktitle}{{\em SIGKDD}}. \bibinfo{pages}{1176--1185}.
\newblock


\bibitem[\protect\citeauthoryear{Sar{\'\i}y{\"u}ce, Gedik, Jacques-Silva, Wu,
  and {\c{C}}ataly{\"u}rek}{Sar{\'\i}y{\"u}ce et~al\mbox{.}}{2013}]%
        {sariyuce2013streaming}
\bibfield{author}{\bibinfo{person}{Ahmet~Erdem Sar{\'\i}y{\"u}ce},
  \bibinfo{person}{Bu{\u{g}}ra Gedik}, \bibinfo{person}{Gabriela
  Jacques-Silva}, \bibinfo{person}{Kun-Lung Wu}, {and}
  \bibinfo{person}{{\"U}mit~V {\c{C}}ataly{\"u}rek}.}
  \bibinfo{year}{2013}\natexlab{}.
\newblock \showarticletitle{Streaming algorithms for k-core decomposition}.
\newblock \bibinfo{journal}{{\em VLDB\/}} \bibinfo{volume}{6},
  \bibinfo{number}{6} (\bibinfo{year}{2013}), \bibinfo{pages}{433--444}.
\newblock


\bibitem[\protect\citeauthoryear{Sozio and Gionis}{Sozio and Gionis}{2010}]%
        {sozio2010community}
\bibfield{author}{\bibinfo{person}{Mauro Sozio} {and}
  \bibinfo{person}{Aristides Gionis}.} \bibinfo{year}{2010}\natexlab{}.
\newblock \showarticletitle{The community-search problem and how to plan a
  successful cocktail party}. In \bibinfo{booktitle}{{\em KDD}}. ACM,
  \bibinfo{pages}{939--948}.
\newblock


\bibitem[\protect\citeauthoryear{Stanton and Kliot}{Stanton and Kliot}{2012}]%
        {stanton2012streaming}
\bibfield{author}{\bibinfo{person}{Isabelle Stanton} {and}
  \bibinfo{person}{Gabriel Kliot}.} \bibinfo{year}{2012}\natexlab{}.
\newblock \showarticletitle{Streaming graph partitioning for large distributed
  graphs}. In \bibinfo{booktitle}{{\em KDD}}. ACM, \bibinfo{pages}{1222--1230}.
\newblock


\bibitem[\protect\citeauthoryear{Tatti and Gionis}{Tatti and Gionis}{2015}]%
        {tatti2015density}
\bibfield{author}{\bibinfo{person}{Nikolaj Tatti} {and}
  \bibinfo{person}{Aristides Gionis}.} \bibinfo{year}{2015}\natexlab{}.
\newblock \showarticletitle{Density-friendly graph decomposition}. In
  \bibinfo{booktitle}{{\em WWW}}. ACM, \bibinfo{pages}{1089--1099}.
\newblock


\bibitem[\protect\citeauthoryear{Tsourakakis, Bonchi, Gionis, Gullo, and
  Tsiarli}{Tsourakakis et~al\mbox{.}}{2013}]%
        {tsourakakis2013denser}
\bibfield{author}{\bibinfo{person}{Charalampos Tsourakakis},
  \bibinfo{person}{Francesco Bonchi}, \bibinfo{person}{Aristides Gionis},
  \bibinfo{person}{Francesco Gullo}, {and} \bibinfo{person}{Maria Tsiarli}.}
  \bibinfo{year}{2013}\natexlab{}.
\newblock \showarticletitle{Denser than the densest subgraph: extracting
  optimal quasi-cliques with quality guarantees}. In \bibinfo{booktitle}{{\em
  KDD}}. ACM, \bibinfo{pages}{104--112}.
\newblock


\bibitem[\protect\citeauthoryear{Valari, Kontaki, and Papadopoulos}{Valari
  et~al\mbox{.}}{2012}]%
        {valari2012discovery}
\bibfield{author}{\bibinfo{person}{Elena Valari}, \bibinfo{person}{Maria
  Kontaki}, {and} \bibinfo{person}{Apostolos~N Papadopoulos}.}
  \bibinfo{year}{2012}\natexlab{}.
\newblock \showarticletitle{Discovery of top-k dense subgraphs in dynamic graph
  collections}. In \bibinfo{booktitle}{{\em SSDBM}}. \bibinfo{pages}{213--230}.
\newblock


\bibitem[\protect\citeauthoryear{Yang and Leskovec}{Yang and Leskovec}{2015}]%
        {yang2015defining}
\bibfield{author}{\bibinfo{person}{Jaewon Yang} {and} \bibinfo{person}{Jure
  Leskovec}.} \bibinfo{year}{2015}\natexlab{}.
\newblock \showarticletitle{Defining and evaluating network communities based
  on ground-truth}.
\newblock \bibinfo{journal}{{\em ICML\/}} \bibinfo{volume}{42},
  \bibinfo{number}{1} (\bibinfo{year}{2015}), \bibinfo{pages}{181--213}.
\newblock


\end{thebibliography}

\end{document}